\RequirePackage{fix-cm}

\documentclass[smallcondensed]{svjour3}

\smartqed 
\def\makeheadbox{\relax} 
\usepackage{graphicx}

\usepackage{amsmath,amssymb,graphicx}

\newtheorem{notation}[theorem]{\bfseries Notation}

\usepackage{url}
\usepackage{siunitx}
\usepackage{bbold}
\usepackage{mathtools}
\usepackage{graphicx}
\usepackage{booktabs}
\usepackage{tikz}
\usepackage{tkz-graph}
\usetikzlibrary{calc}
\usepackage{subcaption}
\usepackage[titletoc]{appendix}
\usepackage{float}
\usepackage{forloop}
\usepackage{ifthen}

\newcommand\peq{\stackrel{\mathclap{\normalfont\mbox{\tiny p}}}{=}}
\newcommand\npeq{\not\peq}
\newcommand\lexsucc{\stackrel{\mathclap{\normalfont\mbox{\tiny $\mathtt{L}$}}}{\succ}}
\newcommand\restr[2]{\ensuremath{\left.#1\right|_{#2}}}
\DeclareMathOperator{\spn}{span}
\DeclareMathOperator{\diag}{diag}
\DeclareMathOperator{\sign}{sgn}
\DeclareMathOperator{\diam}{diam}
\DeclareMathOperator{\dist}{dist}

\usepackage{xcolor,colortbl}
\newcommand{\mc}[2]{\multicolumn{#1}{c}{#2}}
\definecolor{Gray}{gray}{0.85}
\definecolor{LightCyan}{rgb}{0.88,1,1}
\newcolumntype{a}{>{\columncolor{Gray}}c}
\newcolumntype{b}{>{\columncolor{white}}c}

\newlength{\Oldarrayrulewidth}

\def\R{\mathbb R}

\def\Z{\mathbb Z}
\def\N{\mathbb N}
\def\1{\mathbb 1}

\begin{document}
\title{A graph isomorphism invariant based on neighborhood aggregation}
\author{Alp\'ar J\"uttner and P\'eter Madarasi}

\institute{
  A. J\"uttner \at Department of Operations Research, ELTE E\"otv\"os Lor\'and University, P\'azm\'any P\'eter s\'et\'any 1/C, 1117 Budapest, Hungary. \email{alpar.juttner@ttk.elte.hu}
\and
P. Madarasi \at Department of Operations Research, ELTE E\"otv\"os Lor\'and University, and the ELKH-ELTE Egerv\'ary Research Group on Combinatorial Optimization, E\"otv\"os Lor\'and Research Network (ELKH), P\'azm\'any P\'eter s\'et\'any 1/C, 1117 Budapest, Hungary. \email{madarasip@staff.elte.hu}
}

\date{}

\maketitle

\begin{abstract}
  This paper presents a new graph isomorphism invariant, called \emph{$\mathfrak{w}$-labeling}, that can be used to design a polynomial-time algorithm for solving the graph isomorphism problem for various graph classes.
  For example, all non-cospec\-tral graph pairs are distinguished by the proposed combinatorial method, furthermore, even non-isomorphic cospectral graphs can be distinguished assuming certain properties of their eigenspaces.

  We also investigate a refinement of the aforementioned labeling, called \emph{$\mathfrak{s}^k$-labeling}, which has both theoretical and practical applications.
  Among others, it can be used to generate graph fingerprints, which uniquely identify all graphs in the considered databases, including all strongly regular graphs on at most 64 nodes and all graphs on at most 12 nodes.
  It provably identifies all trees and 3-connected planar graphs up to isomorphism, which --- as a byproduct --- gives a new isomorphism algorithm for both graph classes.
  The practical importance of this fingerprint lies in significantly speeding up searching in graph databases, which is a commonly required task in biological and chemical applications.
\end{abstract}

\keywords{Graph isomorphism, Graph fingerprint, Graph hash, Graph databases, Strongly regular graphs, Isomorphism invariant, Planar graph}

\section{Introduction}

The \emph{graph isomorphism problem} is one of the few natural problems in NP that are neither known to be in P nor NP-complete.
At the same time, polynomial-time graph isomorphism algorithms have been developed for various graph classes, like trees and planar graphs~\cite{PlanarGraphIso}, bounded valence graphs~\cite{BondedDegGraphIso}, interval graphs~\cite{IntervalGraphIso} or permutation graphs~\cite{PermGraphIso}.
Furthermore, an FPT algorithm has recently been presented for the \emph{colored hypergraph isomorphism problem}~\cite{ColoredHiperGraphIso}.
The graph isomorphism problem has been extensively studied from a practical point of view, and it can be solved efficiently in most situations.
The best practical graph isomorphism algorithms include Nauty~\cite{Nauty}, VF2~\cite{VF2} and its variants~\cite{VF2pp}.

In many applications, not only two graphs are to be compared, but an isomorphic copy of a given graph $G$ is to be found in a large graph database.
Instead of solving the graph isomorphism problem between $G$ and each graph in the database, one might generate so-called graph fingerprints which satisfy the following property: if the fingerprints of two graphs are different, then the graphs are not isomorphic.
After computing the fingerprints for all graphs in the database as a preprocessing step and also for the query graph $G$, one can filter the graphs in the database which has the same fingerprint as $G$.
Then, we can check whether any of these filtered graphs are isomorphic to $G$.

\paragraph{\normalfont\textbf{Previous work}}
Graph fingerprints are widely used, and multiple schemes have been proposed to generate them.
For example, graph fingerprints were generated by considering the (node labels of) short paths in~\cite{Shasha}.
The spectrum of (the adjacency matrix of) a graph is another isomorphism invariant, which has been studied from a theoretical point of view~\cite{Vandam,Wilson}, and they were also combined with so-called heat-kernels~\cite{Raviv}, which made them also practically applicable.
The number of graphs with unique spectrum was numerically examined up to 12 nodes in~\cite{BrouwerSpence}, and around 80\% of the graphs were found to be determined by their spectrum.

Recently, various algorithms have been developed based on discrete-time quantum walks (DTQW) or continuous-time quantum walks (CTQW), aiming at distinguishing non-isomorphic graph pairs.
It is well known that neither standard single-particle DTQW nor CTQW can distinguish strongly regular graphs (SRG) with the same parameters, furthermore, a constant-particle CTQW without interaction can distinguish no SRG pairs of the same parameters, see~\cite{BrendanDouglas} and~\cite{RudingerKenneth}.
However, the distinguishing power of a variant of single-particle DTQW presented in~\cite{BrendanDouglas} turned out to be larger than that of a standard DTQW.
Namely, it generates different signatures for certain non-isomorphic SRG pairs with the same parameters, but there are still SRG pairs that it fails to distinguish.
In~\cite{PhaseModifiedCTQW}, CTQW were shown to be less powerful than DTQW as far as the graph isomorphism is concerned.
On the other hand, a state-of-the-art quantum walk method using interacting bosons turned out to distinguish all SRG's on at most 64 nodes~\cite{GambleFriesen}.
This compares to the easy-to-compute fingerprint introduced in Section~\ref{sec:gi:StructureOfWalks}, which distinguishes all the mentioned SRG's and, in addition, it also provides a compact fingerprint of the graphs.

Color refinement is a classical method used to show that two graphs are not isomorphic.
Although color refinement does not succeed on all graphs pairs, the method distinguishes a random graph from any non-isomorphic graphs with high probability~\cite{babaiErdos1980random}.
One possible extension of color refinement is the celebrated Weisfeiler-Leman method~\cite{weisfeiler1968reduction}, also called multidimensional color refinement.
In this paper, we consider another procedure, which can be seen as a different way to generalize the classical color refinement algorithm.

\paragraph{\normalfont\textbf{Our results}}
The present paper introduces the concept of $\mathfrak{w}$-labeling, which can be used to solve the graph isomorphism problem in polynomial time under certain conditions --- which hold for a wide range of the graph pairs.
All non-cospectral graph pairs are proved to be distinguished by the proposed combinatorial method --- without computing the graph spectra.
Furthermore, even if the graphs are cospectral, various conditions are shown which ensure that the graphs are distinguished.

A refinement of the aforementioned labeling called $\mathfrak{s}^k$-labeling is also introduced.
Its applications include a fingerprint generation method, which uniquely identifies all the graphs in the considered graph databases --- including all known strongly regular graphs and all graphs on at most 12 nodes.
Therefore, it is competitive with the state-of-the-art quantum walk algorithms.
In addition, it compresses all information about the graph to a short isomorphism-invariant fingerprint.
We also construct two non-isomorphic graphs which it cannot distinguish.

\paragraph{\normalfont\textbf{Notation}}
As usual, sets are described in curly brackets, and multisets are described in curly brackets followed by a superscript hash character.
For example, $\{1,2\}$ denotes the set consisting of the numbers $1,2$, and $\{1,1,2\}^\#$ denotes the multiset consisting of numbers $1,1$ and $2$.
Let $\N$ and $\Z_+$ denote the set of positive and non-negative integers, respectively.

Throughout this paper, $G=(V,E),\ G_1=(V_1,E_1)$ and $G_2=(V_2,E_2)$ denote three arbitrary loop-free, undirected graphs with at least two nodes, where $V,V_1,V_2$ are the node sets and $E,E_1,E_2$ are the edge sets, respectively.
For the sake of simplicity, all three node sets are assumed to be $\{1,\dots,n\}$, that is $V=V_1=V_2=\{1,\dots,n\}$.
The adjacency matrices of these graphs are $A,A_1,A_2\in\{0,1\}^{n\times n}$, respectively.
Let $N_G(i)$ denote the set of the neighbors of node $i$ in graph $G$.

Unless stated otherwise, the presented results apply to graphs having loops, as well.
Note that in the case of node-labeled graphs, the labels can be modeled by adding loops, and clearly, even if the graph has both loops and node labels, there is a compact way to encode them using loops only.

Let $\lambda_1\geq \lambda_2 \geq \dots \geq\lambda_n$ and $\mu_1\geq \mu_2 \geq \dots \geq\mu_n$ denote the eigenvalues of $A_1$ and $A_2$, respectively.
$G_1$ and $G_2$ are \emph{cospectral} if $\lambda_i = \mu_i$ for all $i$.
Let $U,V\in \R^{n\times n}$ be orthogonal matrices (that is, $U^TU=I$ and $V^TV=I$) such that $A_1U=U\diag(\lambda_1, \lambda_2 , \dots ,\lambda_n)$ and $A_2V=V\diag(\mu_1, \mu_2 , \dots ,\mu_n)$.
$U$ and $V$ are called the \emph{eigenmatrices} of $G_1$ and $G_2$, respectively.
Let $u_1,u_2,\dots,u_n$ and $v_1,v_2,\dots,v_n$ denote the column vectors of $U$ and $V$, respectively.
Note that $V$ denotes both the eigenmatrix of $G_2$ and the node set of $G$, but this will not cause ambiguity.
Let $u_{ij}$ denote the $j^{\text{th}}$ entry of the eigenvector $u_i$, in other words, it is the entry of $U$ in the $j^{\text{th}}$ row and $i^{\text{th}}$ column, where $i,j\in \{1,\dots,n\}$.
The first $k$ columns of a matrix $Q$ are denoted by $\restr{Q}{k}$.
Finally, let $\delta_{ij}=1$ if $i=j$, and $0$ otherwise.

\section{Labeling by Counting the Walks}\label{sec:gi:CountingWalks}
For a graph $G=(V,E)$, let $\mathfrak{w}_G : V\to \Z_+^{V\times \Z_+}$ be such that $\mathfrak{w}_G(i)_{jl}$ denotes the number of walks of length $l$ between node $i$ and node $j$ for $l\geq 0$.
The function $\mathfrak{w}_G$ will be referred to as \emph{(infinite) $\mathfrak{w}$-labeling}.
Two matrices $Q_1$ and $Q_2$ are said to be \emph{permutation-equal} if there exists a permutation matrix $P$ for which $PQ_1=Q_2$. This equivalence relation is denoted by $Q_1\peq Q_2$.

The following claim easily follows by the definition of $\mathfrak{w}_G$.
\begin{claim}
  If $\mathfrak{w}(u)\npeq\mathfrak{w}(v)$ for two nodes $u\in V_1$ and $v\in V_2$, then there is no isomorphism between $G_1$ and $G_2$ that maps node $u$ to node $v$.
\end{claim}

\begin{definition}
  $G_1$ and $G_2$ are \emph{$\mathfrak{w}$-equivalent} if the nodes can be relabeled such that $\mathfrak{w}_{G_1}(i)\peq \mathfrak{w}_{G_2}(i)$ for each node $i$.
\end{definition}

\begin{claim}
  If two graphs are isomorphic, then they are $\mathfrak{w}$-equivalent.
\end{claim}

Later on, it will be shown that the reverse direction holds for special graph pairs.

\subsection{Only Short Walks Matter}

The matrices that $\mathfrak{w}$ assigns to the nodes are infinite long, therefore it is not straightforward to check whether two such matrices are permutation-equal or not.
Now, we prove that it suffices to compare the first $(n+1)$ columns of the matrices.

\begin{definition}
  For given column vectors $q_0,q_1,\dots$ over a field, $\spn(q_0,q_1,\dots)$ denotes the linear subspace generated by column vectors $q_0,q_1,\dots$.
\end{definition}

The following lemma will be useful in the proof of Theorem~\ref{thm:gi:onlyshortwalks}.

\begin{lemma}\label{lem:gi:spanlemma}
  For an arbitrary real square matrix $M\in\R^{n\times n}$ and a column vector $q_0\in\R^n$, $\spn(q_0,q_1,q_2,\dots)=\spn(q_0,q_1,\dots,q_{n-1})$, where $q_i:=M^{i}q_0$ for all $i\geq 0$.
\end{lemma}
\begin{proof}
  By induction, one may show that $\spn(q_0,q_1,\dots,q_i)=\spn(q_0,q_1,\dots,q_{i+1})$ implies that $\spn(q_0,q_1,\dots,q_i)=\spn(q_0,q_1,q_2,\dots)$ for all $i$.
  Therefore, the columns $q_0,q_1,\dots,q_n$ generate $\spn(q_0,q_1,q_2,\dots)$.
\end{proof}

The following theorem shows that it is sufficient to consider the first $(n+1)$ columns of the $\mathfrak{w}$-labels, that is, only the number of short walks matters.
Let matrix $\restr{\mathfrak{w}_G}{k}(i)$ denote the first $k$ columns of matrix $\mathfrak{w}_G(i)$.

\begin{theorem}\label{thm:gi:onlyshortwalks}
  For every graph pair $G_1, G_2$ on $n$ nodes and for all $i_1\in V_1,i_2\in V_2$, $\mathfrak{w}_{G_1}(i_1)\peq \mathfrak{w}_{G_2}(i_2)$ if and only if $\restr{\mathfrak{w}_{G_1}}{n+1}(i_1)\peq \restr{\mathfrak{w}_{G_2}}{n+1}(i_2)$.
\end{theorem}
\begin{proof}
  In this proof, let $Q_1, Q_2, Q_1'$ and $Q_2'$ denote the matrices $\mathfrak{w}_{G_1}(v_1), \mathfrak{w}_{G_2}(v_2),$ $\restr{\mathfrak{w}_{G_1}}{n+1}(v_1)$ and $\restr{\mathfrak{w}_{G_2}}{n+1}(v_2)$, respectively.
  If $Q_1\peq Q_2$, then, by definition, there exists a permutation matrix $P$ for which $PQ_1=Q_2$.
  Clearly, $PQ_1=Q_2$ implies that $PQ_1'=Q_2'$.
  To show the other direction, suppose that $Q_1'\peq Q_2'$, and let $q_0,q_1,q_2,\dots$ and $q_0',q_1',q_2',\dots$ denote the columns of $Q_1$ and $Q_2$, respectively.
  Recall that $A_1$ and $ A_2$ denote the adjacency matrices of $G_1$ and $G_2$, respectively.
  Since $Q_1'\peq Q_2'$, there exists a permutation matrix $P$ for which $PQ_1'= Q_2'$, thus it is sufficient to prove that $Pq_i=q_i'$ hols for all $i\geq n+1$.\\
  By induction, suppose that $Pq_{k}=q_{k}'$ holds for all $0\leq k<i$.
  By Lemma~\ref{lem:gi:spanlemma}, there exist coefficients $\alpha_0,\dots\alpha_{n-1}\in\R$ for which $q_{i-1}=\sum\limits_{j=0}^{n-1}\alpha_jq_{j}$ and $q_{i-1}'=\sum\limits_{j=0}^{n-1}\alpha_jq_{j}'$.
  Therefore,
  \begin{multline}
    Pq_i = PA_1q_{i-1} = P\sum\limits_{j=0}^{n-1}\alpha_jA_1q_{j} = \sum\limits_{j=0}^{n-1}\alpha_jPq_{j+1}=\sum\limits_{j=0}^{n-1}\alpha_jq_{j+1}'=\sum\limits_{j=0}^{n-1}\alpha_jA_2q_{j}'\\
    = A_2q_{i-1}'=q_i'
  \end{multline}
  holds for all $i\geq n+1$, which had to be shown.
\end{proof}

The following example shows that the previous theorem is tight in the sense that it is not always sufficient to consider only the first $n$ columns of the $\mathfrak{w}$-labels.

\begin{example}\label{ex:gi:nPlus1ColumnsExample}
  Let $P_n$ denote the path of $n$ nodes, and let $P_n'$ denote the path of $n$ nodes with a loop on one of its endpoints.
  To distinguish two loop-free endpoints of the two graphs, we need to consider the first $(n+1)$ columns of the $\mathfrak{w}$-labels, since their labels do not turn out to be different earlier.
\end{example}

Note that Theorem~\ref{thm:gi:onlyshortwalks} holds even in the following stronger form, which gives the number of necessary columns in terms of the rank of the $\mathfrak{w}$-labels.

\begin{theorem}\label{thm:gi:onlyShortWalks2}
  For every graph pair $G_1, G_2$ and for all $i_1\in V_1,i_2\in V_2$, $\mathfrak{w}_{G_1}(i_1)\peq \mathfrak{w}_{G_2}(i_2)$ if and only if $\restr{\mathfrak{w}_{G_1}}{s(i_1)}(i_1)\peq \restr{\mathfrak{w}_{G_2}}{s(i_2)}(i_2)$, where $s(i_1)=r(\mathfrak{w}_{G_1}(i_1))+1$ and $s(i_2)=r(\mathfrak{w}_{G_2}(i_2))+1$.
\end{theorem}

The proof is similar to that of Theorem~\ref{thm:gi:onlyshortwalks}, therefore it is omitted.
Combining Lemma~\ref{lem:gi:spanlemma} and Theorem~\ref{thm:gi:onlyShortWalks2}, one gets that it is sufficient to generate the columns of the $\mathfrak{w}$-labels one by one and stop as soon as the current column is linearly dependent from the previous columns.
The following theorem gives an upper bound on the largest rank of the $\mathfrak{w}$-labels --- and hence implies an upper bound on the number of columns to be computed.
First, consider the following notations.
Let $\diam(G,i)$ denote the longest shortest path starting from node $i$, formally, $\diam(G,i):=\max\{\dist(i,j) : j\in V_G\}$, where $\dist(i,j)$ is the distance of nodes $i$ and $j$ in $G$.
Let $R, p$ and $\diam(G)$ denote the largest rank of the node labels, the number of distinct eigenvalues and the diameter of $G$, respectively.

\begin{theorem}
  If $G$ is connected, then $p\geq R\geq \diam(G)+1$.
\end{theorem}
\begin{proof}
  Let $Q$ denote a node label having the largest rank, that is, $r(Q)=R$.
  By Lemma~\ref{lem:gi:spanlemma}, the first $R$ columns of $Q$ are linearly independent, which implies that matrices $I,A,A^2,\dots,A^{R-1}$ are linearly independent as well.

  It is well-known that the minimal polynomial of a real symmetric matrix $A$ is $m_A(x)=\prod\limits_{i=1}^p(x-\tilde\lambda_i)$, where $\tilde\lambda_1,\tilde\lambda_2,\dots,\tilde\lambda_p$ are the distinct eigenvalues of $A$.
  Therefore, $p=\deg(m_A)$, and hence $I,A,A^2,\dots,A^{p}$ are linearly dependent, which implies that $p\geq R$.

  To prove that $R\geq \diam(G)+1$, observe that $r(\mathfrak{w}_{G}(i))>\diam(G,i)$ at any node $i$, that is, the rank of $\mathfrak{w}_{G}(i)$ is larger than the length of the longest shortest path from node $i$.
  Applying this to a node $i$ that realizes the diameter of the graph (that is, $\diam(G)=\diam(G,i)$), one gets that $R\geq r(\mathfrak{w}_{G}(i))>\diam(G,i)=\diam(G)$.
\end{proof}

From now on, $\mathfrak{w}_{G}$ might refer to $\restr{\mathfrak{w}_{G}}{n+1}$ or the infinite $\mathfrak{w}$-labeling.
Note that the label $\restr{\mathfrak{w}_{G}}{n+1}(i)$ of a given node $i$ can be computed in $\mathcal{O}(nm)$ operations using a straightforward dynamic programming method.
Furthermore, one might prove that the occurring numbers consist of polynomially many bits in the size of the graph.
Therefore, it takes $\mathcal{O}(n^2m+n^3\log(n))$ steps to decide whether two graphs are $\mathfrak{w}$-equivalent by sorting the $\mathfrak{w}$-labels of both graphs.

\subsection{Spectral Results}\label{sec:gi:Spectral}

In this section, we investigate the relationship between $\mathfrak{w}$-equivalence and the spectra or eigenspaces of graphs.
The following technical observations will be useful when we prove that non-cospectral graph pairs are not $\mathfrak{w}$-equivalent in Theorem~\ref{thm:gi:diffSpectrumDiffWalks}.
\begin{claim}\label{cl:gi:numOfClosedWalksEqu}
  If $\mathfrak{w}_{G_1}(i)\peq \mathfrak{w}_{G_2}(i')$, then the number of closed walks of length $l$ starting from $i\in V_1$ and $i'\in V_2$ are the same for all $l\geq 0$.
\end{claim}
\begin{proof}
  By definition, there exists a permutation matrix $P$ for which $P\mathfrak{w}_{G_1}(i)=\mathfrak{w}_{G_2}(i')$.
  Notice that the first column of $\mathfrak{w}_{G_1}(i)$ and $\mathfrak{w}_{G_2}(i')$ enforces that $P$ maps the $i^{\text{th}}$ row of $\mathfrak{w}_{G_1}(i)$ to the $i'^{\text{th}}$ row of $\mathfrak{w}_{G_2}(i')$, which means that the number of closed walks from $i\in V_1$ and $i'\in V_2$ are the same for all $l\geq 0$.
\end{proof}

Recall that $U\in\R^{n\times n}$ is the eigenmatrix of $G$, and $u_{ij}$ denotes the $j^{\text{th}}$ entry of the $i^{\text{th}}$ column of $U$.

\begin{lemma}\label{lem:gi:eigenDecomp}
  For all $i,j\in \{1,\dots,n\}$ and for all $l\geq 0$, $(A^l)_{ij}=\sum\limits_{k=1}^{n}u_{ki}u_{kj}\lambda_k^l$ holds, where $\lambda_1,\lambda_2, \dots, \lambda_n$ are the eigenvalues of $G$. The right-hand side of this equation will be referred to as the \emph{eigen decomposition}.
\end{lemma}
\begin{proof}
  By definition, $U$ is an orthonormal and $AU=U\diag(\lambda_1,\lambda_2, \dots, \lambda_n)$. Clearly, $A=U\diag(\lambda_1,\lambda_2, \dots, \lambda_n)U^{-1}$ holds, hence $A^l=U\diag(\lambda_1^l,\lambda_2^l, \dots, \lambda_n^l)U^{-1}$. Therefore, $(A^l)_{ij}=\sum\limits_{k=1}^{n}u_{ki}u_{kj}\lambda_k^l$ for any node pair $i,j\in \{1,\dots,n\}$.
\end{proof}

The following observation is an immediate consequence of this lemma.
\begin{corollary}\label{lem:gi:eigenDecompDist}
  For all $i,j\in \{1,\dots,n\}$ and $l\geq 0$, there exist $\beta_1^{ij}, \beta_2^{ij},
  \dots, \beta_p^{ij}\in\R$ for which $(A^l)_{ij}=\sum\limits_{m=1}^{p}\beta_m^{ij}\tilde\lambda_m^l$, where $\tilde\lambda_1,\tilde\lambda_2, \dots, \tilde\lambda_p$ are the distinct non-zero eigenvalues of $G$. The right-hand side of this equation will be referred to as the \emph{aggregated eigen decomposition}.
\end{corollary}
\begin{proof}
  By Lemma~\ref{lem:gi:eigenDecomp}, $(A^l)_{ij}=\sum\limits_{k=1}^{n}u_{ki}u_{kj}\lambda_k^l$ for $l\geq 0$ and $i,j\in \{1,\dots,n\}$. Clearly, $\beta_m^{ij}:=\sum\limits_{k:\lambda_k=\tilde\lambda_m}u_{ki}u_{kj}$ is a proper choice, where $i,j\in \{1,\dots,n\}$ and $m \in \{1,\dots,p\}$.
\end{proof}

The following theorem shows that non-cospectral graphs are not $\mathfrak{w}$-equivalent.
\begin{theorem}\label{thm:gi:diffSpectrumDiffWalks}
  If $G_1$ and $G_2$ are $\mathfrak{w}$-equivalent, then the spectra of $G_1$ and $G_2$ are the same.
\end{theorem}
\begin{proof}
  The proof consists of two steps.

  \textbf{Step 1:} We prove that the set of non-zero eigenvalues of $G_1$ and $G_2$ are the same.

  \begin{lemma}\label{lem:gi:useMutualEigenValsOnly}
    Coefficient $\beta_k^{ii}$ in the aggregated eigen decomposition is zero if it corresponds to a non-zero eigenvalue of exactly one of $G_1$ and $G_2$ for all $i,k\in \{1,\dots,n\}$.
  \end{lemma}
  \begin{proof}
    Let $\tilde\lambda_1,\tilde\lambda_2,\dots,\tilde\lambda_r,\tilde\theta_{r+1},\dots,\tilde\theta_{p}$ and $\tilde\lambda_1,\tilde\lambda_2,\dots,\tilde\lambda_r,\tilde\mu_{r+1},\dots,\tilde\mu_{q}$ denote all the distinct non-zero eigenvalues of $G_1$ and $G_2$, respectively, where $\tilde\lambda_1,\tilde\lambda_2,\dots,\tilde\lambda_r$ are the mutual non-zero eigenvalues of the two graphs and $\tilde\theta_{r+1},\dots,\tilde\theta_{p},\tilde\mu_{r+1},\dots,\tilde\mu_{q}$ are pairwise distinct.

    For the sake of simplicity, suppose that the nodes are re-indexed in such a way that the identity mapping gives a $\mathfrak{w}$-equivalence, that is, $\mathfrak{w}_{G_1}(i)\peq \mathfrak{w}_{G_2}(i)$ for each node $i$.

    For any $i$ and $j$, there exist
    coefficients
    $\alpha_1,\alpha_2,\dots,\alpha_p,\allowbreak\beta_1,\beta_2,\dots,\beta_q$ such that
    \begin{equation}
      (A_1^l)_{ij}=\sum\limits_{k=1}^{r}\alpha_k\tilde\lambda_k^l
      +\sum\limits_{k=r+1}^{p}\alpha_k\tilde\theta_k^l
    \end{equation}
    and
    \begin{equation}
      (A_2^l)_{ij}=\sum\limits_{k=1}^{r}\beta_k\tilde\lambda_k^l
      +\sum\limits_{k=r+1}^{q}\beta_k\tilde\mu_k^l
    \end{equation}
    hold for all $l\geq 1$, by Corollary~\ref{lem:gi:eigenDecompDist}.
    The two graphs being $\mathfrak{w}$-equivalent, one gets that
    \begin{equation*}
      \sum\limits_{k=1}^{r}\alpha_k\tilde\lambda_k^l
      +\sum\limits_{k=r+1}^{p}\alpha_k\tilde\theta_k^l=(A_1^l)_{ii}=(A_2^l)_{ii}=\sum\limits_{k=1}^{r}\beta_k\tilde\lambda_k^l
      +\sum\limits_{k=r+1}^{q}\beta_k\tilde\mu_k^l
    \end{equation*}
    holds for all $i\in \{1,\dots,n\}$ and $l\geq 1$, where the second equation follows from Claim~\ref{cl:gi:numOfClosedWalksEqu}.
    Subtracting the right-hand side, one obtains that
    \begin{equation}
      \sum\limits_{k=1}^{r}(\alpha_k-\beta_k)\tilde\lambda_k^l
      +\sum\limits_{k=r+1}^{p}\alpha_k\tilde\theta_k^l-\sum\limits_{k=r+1}^{q}\beta_k\tilde\mu_k^l=0
    \end{equation}
    for all $l\geq 1$.
    Let $m:=p+q-r$, and consider the following linear equations for $l \in \{1,\dots,m\}$.
    \begin{equation}
      \sum\limits_{k=1}^{r}x_k\tilde\lambda_k^l
      +\sum\limits_{k=r+1}^{p}x_k\tilde\theta_k^l+\sum\limits_{k=r+1}^{q}x_{p+k-r}\tilde\mu_k^l=0,
    \end{equation}
    where
    \begin{equation}
      x_s:=
      \begin{cases}
        \alpha_s-\beta_s& \text{if } 1\leq s\leq r,\\
        \alpha_s        & \text{if } r+1\leq s\leq p,\\
        -\beta_{r+s-p}  & \text{if } p+1\leq s\leq p+q-r,\\
      \end{cases}
    \end{equation} for all $s\in \{1,\dots,m\}$.
    The matrix of this linear equation system is
    \begin{equation}
      M:=\begin{bmatrix}
        \tilde\lambda_1^1 & \dots & \tilde\lambda_r^1 & \tilde\theta_{r+1}^1 &\dots & \tilde\theta_{p}^1 & \tilde\mu_{r+1}^1 &\dots & \tilde\mu_{q}^1 \\
        \tilde\lambda_1^2 & \dots & \tilde\lambda_r^2 & \tilde\theta_{r+1}^2 &\dots & \tilde\theta_{p}^2 & \tilde\mu_{r+1}^2 &\dots & \tilde\mu_{q}^2 \\
        \tilde\lambda_1^3 & \dots & \tilde\lambda_r^3  & \tilde\theta_{r+1}^3 &\dots & \tilde\theta_{p}^3 & \tilde\mu_{r+1}^3 &\dots & \tilde\mu_{q}^3 \\
        \vdots & \ddots & \vdots & \vdots & \ddots & \vdots & \vdots & \ddots & \vdots \\
        \tilde\lambda_1^m & \dots & \tilde\lambda_r^m  & \tilde\theta_{r+1}^m &\dots & \tilde\theta_{p}^m & \tilde\mu_{r+1}^m &\dots & \tilde\mu_{q}^m \\
      \end{bmatrix}.
    \end{equation}

    Observe that $M=M'\diag(\tilde\lambda_1^1,  \dots,  \tilde\lambda_r^1, \tilde\theta_{r+1}^1, \dots ,\tilde\theta_{p}^1, \tilde\mu_{r+1}^1, \dots,  \tilde\mu_{q}^1)$, where $M'$ denotes the following Vandermonde matrix.
    \begin{equation}
      M':=\begin{bmatrix}
        1 & \dots & 1 & 1 & \dots & 1 & 1 & \dots & 1 \\
        \tilde\lambda_1^1 & \dots & \tilde\lambda_r^1 & \tilde\theta_{r+1}^1 &\dots & \tilde\theta_{p}^1 & \tilde\mu_{r+1}^1 &\dots & \tilde\mu_{q}^1 \\
        \tilde\lambda_1^2 & \dots & \tilde\lambda_r^2 & \tilde\theta_{r+1}^2 &\dots & \tilde\theta_{p}^2 & \tilde\mu_{r+1}^2 &\dots & \tilde\mu_{q}^2 \\
        \vdots & \ddots & \vdots & \vdots & \ddots & \vdots & \vdots & \ddots & \vdots \\
        \tilde\lambda_1^m & \dots & \tilde\lambda_r^m  & \tilde\theta_{r+1}^m &\dots & \tilde\theta_{p}^m & \tilde\mu_{r+1}^m &\dots & \tilde\mu_{q}^m \\
      \end{bmatrix}
    \end{equation}
    \\
    Therefore, $\det(M)=\det(M')
    \prod\limits_{k=1}^r\tilde\lambda_k
    \prod\limits_{k=r+1}^p\tilde\theta_k
    \prod\limits_{k=r+1}^q\tilde\mu_k\neq 0$, thus the only solution is $x\equiv 0$, that is,
    \begin{equation}
      \begin{cases}
        \alpha_s=\beta_s& \text{if } 1\leq s\leq r,\\
        \alpha_s=0      & \text{if } r+1\leq s\leq p,\\
        \beta_{r+s-p}=0 & \text{if } p+1\leq s\leq p+q-r\\
      \end{cases}
    \end{equation}
    follows for all $s\in \{1,\dots,m\}$.
  \end{proof}
  Let $\lambda^*\neq 0$ denote an eigenvalue which corresponds to exactly one of the graphs, say to $G_1$. Next, we argue that there exists a node $i\in V_1$ such that $\lambda^*$ has a non-zero coefficient in the aggregated eigen decomposition given by Corollary~\ref{lem:gi:eigenDecompDist} for $(A_1^l)_{ii}$ --- contradicting Lemma~\ref{lem:gi:useMutualEigenValsOnly}. Let $\tilde m$ denote the unique index for which $\tilde\lambda_{\tilde m}=\lambda^*$. By Corollary~\ref{lem:gi:eigenDecompDist}, the coefficient of $\tilde\lambda_{\tilde m}$ in the case of the number of closed walks from node $i$ is $\beta^{ii}_{\tilde m}=\sum_{k:\lambda_k=\tilde\lambda_{\tilde m}}u_{ki}u_{ki}$. Let $m$ be an index such that\ $\lambda_m=\tilde\lambda_{\tilde m}$, and let $i$ be such that $u_{mi}u_{mi}>0$ (there exists at least one such index, since $u_mu_m=1$). Observe that $\beta^{ii}_{\tilde m}\geq u_{mi}u_{mi}>0$ holds, therefore node $i$ meets the requirements, contradicting Lemma~\ref{lem:gi:useMutualEigenValsOnly}.\\

  \textbf{Step 2:} We show that the multiplicities of the eigenvalues are the same in $G_1$ and $G_2$.
  It is sufficient to show that the multiplicities of the non-zero eigenvalues are the same, because this also implies that the multiplicities of zero are the same in $G_1$ and $G_2$.
  Let $\tau^{(k)}_i$ denote the multiplicity of $\tilde\lambda_i$ in $G_k$ $(k=1,2)$, where $\tilde\lambda_1,\dots,\tilde\lambda_p$ are the mutual eigenvalues of $G_1$ and $G_2$.

  As a consequence of Lemma~\ref{lem:gi:eigenDecomp}, the sum of the numbers of closed walks of $G_k$ of length $l$ is $\sum\limits_{j=1}^p\tau^{(k)}_j\tilde\lambda_j^l$, $(l\geq 1)$. Since $G_1$ and $G_2$ are $\mathfrak{w}$-equivalent, Claim~\ref{cl:gi:numOfClosedWalksEqu} applies, thus the sum of the numbers of closed walks of length $l$ in the two graphs are the same for all $l$, that is, $\sum\limits_{j=1}^p\tau^{(1)}_j\tilde\lambda_j^l=\sum\limits_{j=1}^p\tau^{(2)}_j\tilde\lambda_j^l$ for all $l\geq 1$. Subtracting the right-hand side, one gets that
  \begin{equation}
    \sum\limits_{j=1}^p(\tau^{(1)}_j-\tau^{(2)}_j)\tilde\lambda_j^l=0
  \end{equation}
  holds for all $l\geq 1$.
  Consider these equations for $l \in \{1,\dots,p\}$, and let $x_j:=\tau^{(1)}_j-\tau^{(2)}_j$ for all $j\in \{1,\dots,p\}$. Similarly to Step 1, the matrix of this equation system has non-zero determinant, thus the only solution is $x\equiv 0$, which means that $\tau^{(1)}_j=\tau^{(2)}_j$ for all $j \in \{1,\dots,p\}$.
  Therefore, each non-zero eigenvalue has the same multiplicities in the two graphs, hence the multiplicities of 0 are the same, as well, which completes the proof.
\end{proof}

Note that $\mathfrak{w}$-equivalence distinguishes more graph pairs than the spectra. For example, a cycle of length $6$ and two disjoint triangles are well-known cospectral graphs, but they are clearly not $\mathfrak{w}$-equivalent.

In what follows, we show that even if two non-isomorphic graphs are cospectral, they may not be $\mathfrak{w}$-equivalent if their eigenspaces are different enough:
\begin{theorem}
  Let $G_1$ and $G_2$ be cospectral graphs with single eigenvalues. If one of the eigenmatrices has a row that contains non-zero elements only, then the two graphs are $\mathfrak{w}$-equivalent if and only if the graphs are isomorphic.
\end{theorem}
\begin{proof}
  Clearly, it suffices to show that if $G_1$ and $G_2$ are $\mathfrak{w}$-equivalent, then they are isomorphic.
  To this end, we show a permutation matrix $\Pi$ such that $\Pi A_1\Pi^T=A_2$.
  Recall that $U=(u_1,\dots,u_n)$ and $V=(v_1,\dots,v_n)$ denote the eigenmatrices of $G_1$ and $G_2$, respectively, that is, $A_1=U\diag(\lambda_1,\dots,\lambda_n)U^T$ and $A_2=V\diag(\lambda_1,\dots,\lambda_n)V^T$. A permutation matrix $\Pi$ corresponds to an isomorphism if and only if $\Pi U\diag(\lambda_1,\dots,\lambda_n)U^T\Pi^T=V\diag(\lambda_1,\dots,\lambda_n)V^T$, which in turn holds if and only if $\Pi U=VS$ for a matrix $S=\diag(\sigma_1,\dots,\sigma_n)$, where $\sigma_i\in \{-1,1\}$. Therefore, it is sufficient to show such matrices $\Pi$ and $S$.

  Without loss of generality, assume that row $i^*$ of $U$ consists of non-zero elements.
  By the definition of $\mathfrak{w}$-equivalence, there is a permutation $\pi$
  such that $(A_1^l)_{i^*j}=(A_2^l)_{\pi(i^*)\pi(j)}$, thus
  $u_{ki^*}u_{kj}=v_{k\pi(i^*)}v_{k\pi(j)}$ for all $j \in \{1,\dots,n\}$. Clearly, row $\pi(i^*)$ of $V$ consists of non-zero elements. Let
  $S:=\diag(\sigma_1,\dots,\sigma_n)$, where
  $\sigma_k:=\sign(u_{ki^*})\sign(v_{k\pi(i^*)})\in\{-1,1\}$, and let
  $\Pi=\begin{cases}
    1 & \text{if } \pi(j)=i,\\
    0 & \text{otherwise.}\\
  \end{cases}$
  In what follows, we argue that $\Pi U = VS$.
  The values in position $(j,k)$ of the left and the right side are $u_{k\pi^{-1}(j)}$ and $\sigma_k v_{kj}$, respectively.
  Observe that $u_{k\pi^{-1}(j)}=\sigma_k v_{kj}$ for all $j,k\in\{1,\dots,n\}$ if and only if $u_{kj}=\sigma_k v_{k\pi(j)}$ for all $j,k\in\{1,\dots,n\}$, which in turn is equivalent to $u_{ki^*}u_{kj}=v_{k\pi(i^*)}v_{k\pi(j)}$ for all $j,k\in\{1,\dots,n\}$, because $u_{ki^*}=\sigma_k v_{k\pi(i^*)}$ and $\sigma_k^2=1$, and $\pi$ was chosen such that $u_{ki^*}u_{kj}=v_{k\pi(i^*)}v_{k\pi(j)}$ for all $j,k\in\{1,\dots,n\}$.
  Hence $\Pi U = VS$, which completes the proof of the theorem.
\end{proof}

The next lemma will be useful in the proofs of Theorem~\ref{thm:gi:diffEigenVecMultisets} and Theorem~\ref{thm:gi:perronEigenvAreDiff}.

\begin{lemma}\label{lem:gi:aboluseValDiff}
  If $G_1$ and $G_2$ are $\mathfrak{w}$-equivalent graphs and the nodes of $G_2$ are re-indexed in such a way that $\mathfrak{w}_{G_1}(i)\peq \mathfrak{w}_{G_2}(i)$ for all $i\in\{1,\dots,n\}$, then for any single eigenvalue, the corresponding normalized eigenvectors in the two graphs are element-wise the same up to sign.
\end{lemma}
\begin{proof}
  By Theorem~\ref{thm:gi:diffSpectrumDiffWalks}, one gets that $G_1$ and $G_2$ are cospectral. By definition and by Claim~\ref{lem:gi:eigenDecomp}, $\restr{\mathfrak{w}_{G_1}}{n+1}(i)\peq \restr{\mathfrak{w}_{G_2}}{n+1}(i)$
  implies that
  \begin{equation}
    \sum\limits_{k=1}^nu_{ki}u_{ki}\lambda_k^l=(A_1^l)_{ii}=(A_2^l)_{ii}
    =\sum\limits_{k=1}^nv_{ki}v_{ki}\lambda_k^l
  \end{equation}
  for all $i\in\{1,\dots,n\}$ and $l\geq 0$, which in turn implies that $u_{ki}u_{ki}=v_{ki}v_{ki}$ holds, where $i\in\{1,\dots,n\}$ is arbitrary and $k$ is such that $\lambda_k$ is a single eigenvalue. That is, $|u_{ki}|=|v_{ki}|$ if $i\in\{1,\dots,n\}$ and $k$ is such that $\lambda_k$ is a single eigenvalue, which completes the proof.
\end{proof}

\begin{theorem}\label{thm:gi:diffEigenVecMultisets}
  Let $G_1$ and $G_2$ be cospectral with single eigenvalues. If $\{u_{ik} : k\in\{1,\dots,n\}\}^\#\neq\{-u_{ik} : k\in\{1,\dots,n\}\}^\#$ and $\{v_{ik} : k\in\{1,\dots,n\}\}^\#\neq\{-v_{ik} : k\in\{1,\dots,n\}\}^\#$ for all $i\in\{1,\dots,n\}$, then the two graphs are $\mathfrak{w}$-equivalent if and only if they are isomorphic.
\end{theorem}
\begin{proof}
  If $G_1$ and $G_2$ are isomorphic, then they are clearly $\mathfrak{w}$-equivalent. To show the other direction, let $G_1$ and $G_2$ be $\mathfrak{w}$-equivalent. Let $w_1,w_2\in\R^n$ be vectors for which $\{w_{1k} : k\in\{1,\dots,n\}\}^\#\neq\{w_{2k} : k\in\{1,\dots,n\}\}^\#$, and let $w_1\lexsucc w_2$ mean that after non-increasingly ordering their coordinates, $w_1$ is lexicographically larger than $w_2$.

  Without loss of generality, one can assume that $u_i\lexsucc-u_i$ and $v_i\lexsucc-v_i$ holds for all $i\in\{1,\dots,n\}$. Let $\pi:V_1\to V_2$ be a $\mathfrak{w}$-equivalence. By Lemma~\ref{lem:gi:aboluseValDiff}, $|u_{ki}|=|v_{k\pi(i)}|$ holds for all $k\in\{1,\dots,n\}$ and $i\in\{1,\dots,n\}$.
  By contradiction, suppose that there is an index $k^*$ and $i^*$ such that $u_{k^*i^*}\neq v_{k^*\pi(i^*)}$. This means that $u_{k^*i^*}=-v_{k^*\pi(i^*)}$. Let $\pi ^*$ denote the bijection of node $i^*$, for which $u_{ki^*}u_{kj}=v_{k\pi ^*(i^*)}v_{k\pi ^*(j)}$ holds for all $j,k\in\{1,\dots,n\}$, regardless of whether 0 is an eigenvalue. Clearly, $\pi^*$ can be prescribed to satisfy $\pi^*(i^*)=\pi(i^*)$. Thus one gets that $u_{k^*i^*}u_{k^*j}=v_{k^*\pi ^*(i^*)}v_{k^*\pi ^*(j)}$ for all $j\in\{1,\dots,n\}$, which implies $-u_{k^*}=\pi v_{k^*}$. But then $u_{k^*}\lexsucc -u_{k^*} = \pi^* v_{k^*}\lexsucc -v_{k^*} = {\pi^*}^{-1}u_{k^*}$, therefore $G_1$ and $G_2$ are indeed isomorphic.
\end{proof}

\begin{definition}
  A graph is \emph{friendly}~\cite{Aflalo2942} if each of its eigenvalues has multiplicity one and $\1 U$ has no zero
  coordinates, where $U$ is the eigenmatrix of the graph.
\end{definition}

\begin{corollary}
  Two friendly graphs $G_1$ and $G_2$ are isomorphic if and only if they are $\mathfrak{w}$-equivalent.
\end{corollary}
\begin{proof}
  The assumptions, $\1 u_i\neq 0$ and $\1 v_i\neq 0$, imply that $\{u_{ik} : k\in\{1,\dots,n\}\}^\#\neq\{-u_{ik} : k\in\{1,\dots,n\}\}^\#$ and $\{v_{ik} : k\in\{1,\dots,n\}\}^\#\neq\{-v_{ik} : k\in\{1,\dots,n\}\}^\#$ for all $i\in\{1,\dots,n\}$, thus Theorem~\ref{thm:gi:diffEigenVecMultisets} can be applied.
\end{proof}

Recall the following well-known theorem.
\begin{theorem}[Perron-Frobenius]\label{thm:gi:PerronFrobEigenVal}
  Let the graph $G$ be connected and have at least two nodes.
  The largest eigenvalue $\lambda_1$ of the adjacency matrix of $G$ is positive, has multiplicity one, and $\lambda_1\geq|\lambda|$ for every eigenvalue $\lambda$. In addition, the eigenvector corresponding to $\lambda_1$ can be chosen strictly positive.
\end{theorem}

The positive normalized eigenvector corresponding to the largest positive eigenvalue in Theorem~\ref{thm:gi:PerronFrobEigenVal} will be referred to as the \emph{Perron-Frobenius eigenvector} of $G$.
The following theorem is an immediate consequence of Lemma~\ref{lem:gi:aboluseValDiff}.

\begin{theorem}\label{thm:gi:perronEigenvAreDiff}
  Let $G_1$ and $G_2$ be connected cospectral graphs on at least two nodes. If the Perron-Frobenius eigenvectors of $G_1$ and $G_2$ are different, then $G_1$ and $G_2$ are not $\mathfrak{w}$-equivalent.

\end{theorem}

The Perron-Frobenius eigenvector of a graph determines the \emph{invariant distribution} with respect to infinite random walks, therefore the previous theorem states that if the invariant distributions of two graphs are different, then they are not $\mathfrak{w}$-equivalent.

\section{Labeling by Perfect Aggregation}\label{sec:gi:StructureOfWalks}
This section introduces a refinement of $\mathfrak{w}$-labeling.
First, let
\begin{equation*}
\mathfrak{s}^1_G(i_1)_{il}:=
\begin{cases}
    (\emptyset,\{\delta_{ii_1}\})                                                     & \text{if } l=0,\\
    (\mathfrak{s}^1_G(i_1)_{i,l-1},\{\mathfrak{s}^1_G(i_1)_{i',l-1}:i'\in N_G(i)\}^\#)& \text{otherwise}\\
\end{cases}
\end{equation*}
for a node $i_1\in V$ and $i\in\{1,\dots,n\}$, $l\in\Z_+$. Essentially, we collect the values of the neighbors into a multiset, instead of adding them together as we did in the case of $\mathfrak{w}$-labels.
Clearly, generating these multisets preserves no less information than adding the values together --- which may give the same sum even if the summands were different.
In fact, we gather strictly more information with $\mathfrak{s}^1$ than with $\mathfrak{w}$, see Figure~\ref{fig:gi:w1IsWeakerThanS1} for an example.

We go even further, and give a generalization of $\mathfrak{s}^1$, which initializes the first iteration, that is the case $l=0$, in a slightly more complicated way.

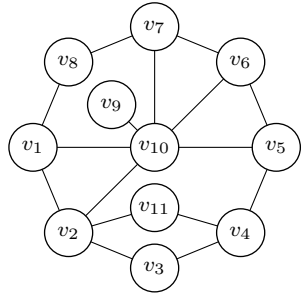
\begin{figure}
  \centering
  \begin{tikzpicture}[scale=.8]
    \SetVertexMath
    \tikzset{VertexStyle/.append style = {minimum size = 18pt,inner sep=0pt}}
    \newcounter{cx}
    \newcounter{prevCx}
    \forLoop{1}{8}{cx}{
      \begin{scope}[rotate=360/8*(\thecx-1)]
        \Vertex[x=-2,y=0,L=v_{\thecx}]{v_\thecx}
      \end{scope}
    }
    \forLoop{2}{8}{cx}{
      \setcounter{prevCx}{
        \numexpr(\thecx)-1\relax
      }
      \draw[] (v_\thecx) -- (v_\theprevCx);
    }
    \Vertex[x=0,y=0,L=v_{10}]{v_10}
    \draw[] (v_1) -- (v_10);
    \draw[] (v_2) -- (v_10);
    \draw[] (v_5) -- (v_10);
    \draw[] (v_6) -- (v_10);
    \draw[] (v_7) -- (v_10);
    \draw[] (v_1) -- (v_8);

    \begin{scope}[scale=.5]
      \begin{scope}[rotate=360/8*(7)]
        \Vertex[x=-2,y=0,L=v_{9}]{v_9}
      \end{scope}
      \begin{scope}[rotate=360/8*(2)]
        \Vertex[x=-2,y=0,L=v_{11}]{v_11}
      \end{scope}
    \end{scope}

    \draw[] (v_9) -- (v_10);
    \draw[] (v_2) -- (v_11);
    \draw[] (v_4) -- (v_11);


  \end{tikzpicture}
  \caption{An example when the $\mathfrak{s}^k$-labeling is strictly stronger than the $\mathfrak{w}$-labeling. Here, $\mathfrak{w}_G(v_1)\peq\mathfrak{w}_G(v_5)$, but $\mathfrak{s}_G^{1}(v_1)\npeq\mathfrak{s}_G^{1}(v_5)$.}\label{fig:gi:w1IsWeakerThanS1}
\end{figure}

\begin{notation}
  Let $\mathfrak{s}^k_G(i_1,\dots,i_k)$ be an $n\times \Z_+$ matrix whose  position $(j,l)$ describes the walks of length at most $l$ between nodes $\{i_1,\dots,i_k\}$ and $i$.
  Formally, define $\mathfrak{s}^k_G(i_1,\dots,i_k)_{il}$ as
\begin{equation}\label{gi:perfAgr:req}
\begin{cases}
    (\emptyset,\{\sum\limits^k_{q=1} q\delta_{ii_q}\})                                                    & \text{if } l=0,\\
    (\mathfrak{s}^k_G(i_1,\dots,i_k)_{i,l-1},\{\mathfrak{s}^k_G(i_1,\dots,i_k)_{i',l-1}:i'\in N_G(i)\}^\#)& \text{otherwise}\\
\end{cases}
\end{equation}
for all $j\in V$ and for all $l\geq 0$.
\end{notation}

Note that the first column of matrix $\mathfrak{s}^k_G(i_1,\dots,i_k)$ corresponds to walks of length zero, therefore its index is zero.
Column $l$ will be denoted by $\mathfrak{s}^k_G(i_1,\dots,i_k)_{\bullet l}$. Recall that $\restr{\mathfrak{s}^k_G(i_1,\dots,i_k)}{q}$ denotes the first $q$ columns of matrix $\mathfrak{s}^k_G(i_1,\dots,i_k)$, that is, it describes the walks up to length $(q-1)$.

Simple inductive proof shows that it suffices to consider the first $(n+1)$ columns, similarly to the
case of $\mathfrak{w}$-labels.
\begin{claim}\label{cl:gi:onlyShortSWLMatter}
Let us given a graph pair $G_1, G_2$, distinct nodes $i_1,\dots,i_k\in V_1,j_1,\dots,j_k\in V_2$ and an integer $k\geq 1$. Then,
\[\mathfrak{s}^k_{G_1}(i_1,\dots,i_k)\peq \mathfrak{s}^k_{G_2}(j_1,\dots,j_k)\]
holds if and only if
\[\restr{\mathfrak{s}^k_{G_1}(i_1,\dots,i_k)}{n+1}\peq \restr{\mathfrak{s}^k_{G_2}(j_1,\dots,j_k)}{n+1},\]
where $n=|V_1|=|V_2|$.
\end{claim}
From now on, $\mathfrak{s}^k_{G}(i_1,\dots,i_k)$ might refer to $\restr{\mathfrak{s}^k_{G}(i_1,\dots,i_k)}{n+1}$, that is, to its first $(n+1)$ columns.
Example~\ref{ex:gi:nPlus1ColumnsExample} shows that the previous claim is tight in the sense that considering the first $n$ columns would not be sufficient.

Note that the size of $\mathfrak{s}^k_G(i_1,\dots,i_k)_{jl}$ may be exponentially large in $n$.
Practically, one may address this issue by hashing the occurring data using SHA512 --- this also enables the generation of short graph fingerprints.
It is also possible to verify in polynomial time whether the $\mathfrak{s}^k$-labels of two given graphs are the same by replacing the labels with small integers as follows.
Let the number associated with the empty set be zero, and that of the set $\{p\}$ be $p$ for all $p\in\{1,\dots,k\}$.
In the recursion~(\ref{gi:perfAgr:req}), we substitute the associated numbers in place of all labels which already have an assigned number.
When a label appears the first time throughout the labeling process, we let its representing number be the next natural number.
Since all labels $\mathfrak{s}^k_G(i_1,\dots,i_k)_{\bullet l-1}$ have an associated number when computing $\mathfrak{s}^k_G(i_1,\dots,i_k)_{\bullet l}$ --- and hence the labels are multisets consisting of small integers --- the process runs in polynomial time.
We emphasize that each occurring label must be associated with one and only one number throughout the whole process, and this number must be used consistently at all occurrences of the label in both graphs.

\begin{notation}\label{nota:gi:potato}
For $q=k-1,\dots,0$, let $\mathfrak{s}^k_G(i_1,\dots,i_{q}) := \{\mathfrak{s}^k_G(i_1,\dots,i_{q},i) : i\in V\setminus\{i_1,\dots,i_{q}\}\}^\#$.
\end{notation}
The purpose of this recursive notation is to derive the isomorphism invariant fingerprint $\mathfrak{s}^k_G$, which is obtained for $q=0$ (we omit the empty parenthesis after $\mathfrak{s}^k_G$). For example, if $k=2$, then one obtains $\mathfrak{s}^k_G(v)=\{\mathfrak{s}^k_G(v,v') : v'\in V\setminus\{v\}\}^\#$ for $q=1$, and the fingerprint $\mathfrak{s}^k_G=\{\mathfrak{s}^k_G(v) : v\in V\}^\#$ for $q=0$.

The following two claims easily follow by definition.
\begin{claim}
For all integer $k\geq 1$ and $q\in\{0,\dots,k-1\}$, if $\mathfrak{s}^k_{G_1}(i_1,\dots,i_q)\neq\mathfrak{s}^k_{G_2}(j_1,\dots,j_q)$ for some nodes $i_1,\dots,i_q\in V_1$ and $j_1,\dots,j_q\in V_2$, then there is no isomorphism between $G_1$ and $G_2$ that maps node $i_r$ to node $j_r$ for each $r\in\{1,\dots,q\}$.
\end{claim}

\begin{claim}\label{cl:gi:distPowerOfSWLIsNonDecr}
For all integer $k\geq 1$, $q\in\{0,\dots,k\}$ and any $i_1,\dots,i_q\in V_1$ and $j_1,\dots,j_q\in V_2$,  if $\mathfrak{s}^{k}_{G_1}(i_1,\dots,i_q)\neq\mathfrak{s}^{k}_{G_2}(j_1,\dots,j_q)$, then $\mathfrak{s}^{k+1}_{G_1}(i_1,\dots,i_q)\neq\mathfrak{s}^{k+1}_{G_2}(j_1,\dots,j_q)$.
\end{claim}

\begin{definition}
Two graphs $G_1$ and $G_2$ are \emph{$\mathfrak{s}^k$-equivalent} if $\mathfrak{s}^{k}_{G_1}=\mathfrak{s}^{k}_{G_2}$.
\end{definition}

As a special case of the previous claim, one gets that if two graphs are not $\mathfrak{s}^k$-equivalent, then they are not $\mathfrak{s}^{k+1}$-equivalent either.
That is, as $k$ is increased, more and more non-isomorphic graph pairs are distinguished.
Note that when $k$ reaches the number of nodes, the graph is uniquely identified by $\mathfrak{s}^k$.

\begin{remark}
For any given constant $k\in\N$, one can verify in polynomial time whether two graphs are $\mathfrak{s}^k$-equivalent or not.
\end{remark}

\subsection{The Distinguishing Power of Perfect Aggregation}

This section investigates the distinguishing power of the above notion on trees, planar graphs, and subject to node-connectivity parameters.

\begin{theorem}
  Two trees are isomorphic if and only if they are $\mathfrak{s}^1$-equivalent.
\end{theorem}
\begin{proof}
  Given two $\mathfrak{s}^1$-equivalent trees $G_1=(V,E_1)$ and $G_2=(V,E_2)$, we show that they are isomorphic.
  For an edge $(r,p)\in E_i$, let $T_i(r,p)=(V_i(r,p),E_i(r,p))$ denote the subtree of $G_i$ obtained as the connected component of $(V,E_i\setminus\{(r,p)\})$ containing node $r$.

  By induction, we prove that for any edges $(r_1,p_1)\in E_1$ and $(r_2,p_2)\in E_2$ if $\restr{\mathfrak{s}^1_{G_1}(r_1)}{n'+1} \peq \restr{\mathfrak{s}^1_{G_2}(r_2)}{n'+1}$ and $\restr{\mathfrak{s}^1_{G_1}(p_1)}{n'+1} \peq \restr{\mathfrak{s}^1_{G_2}(p_2)}{n'+1}$ for $n'=|V_1(r_1,p_1)|$, then $T_1(r_1,p_1)$ and $T_2(r_2,p_2)$ are isomorphic.
  Clearly, if $n'=1$ --- in other words, $r_1$ is a leaf node in $G_1$ --- then $r_2$ must also be a leaf node in $G_2$.

  Otherwise, one gets that $\{ \restr{\mathfrak{s}^1_{G_1}(i)}{n'} : i\in N_{G_1}(r_1)\}^\# = \{ \restr{\mathfrak{s}^1_{G_2}(i)}{n'} : i\in N_{G_2}(r_2)\}^\#$.
  Thus $r_1$ and $r_2$ have the same number of neighbors and there is a one-to-one mapping $\phi: N_{G_1}(r_1) \to N_{G_2}(r_2)$ so that $v$ and $\phi(v)$ have same label up to the first $n'$ columns for each $v\in N_{G_1}(r_1)$.
  Therefore, from the induction hypothesis, $T_1(v,r_1)$ and $T_2(\phi(v),r_2)$ are isomorphic subtrees for all $v\in N_{G_1}(r_1)\setminus p_1$.
  The isomorphism of $T_1(r_1,p_1)$ and $T_2(r_2,p_2)$ follows from this immediately.

  In order to complete the proof of the theorem, let us choose an arbitrary leaf node $r_1\in V_1$ and a node $r_2\in V_2$ with $\mathfrak{s}^1_{G_1}(r_1)\peq\mathfrak{s}^1_{G_2}(r_2)$.
  Node $r_2$ is also a leaf node and $\restr{\mathfrak{s}^1_{G_1}(p_1)}{|V_1|-1}\peq\restr{\mathfrak{s}^1_{G_2}(p_2)}{|V_1|-1}$ for their neighbors $p_1\in V_1$ and $p_2\in V_2$.
  Applying the above claim to $r_1,p_1,r_2,p_2$ proves the isomorphism of $G_1$ and $G_2$.
\end{proof}

Note that this proof provides a new polynomial-time isomorphism algorithm for trees, however, a more efficient algorithm exists~\cite{PlanarGraphIso}.

\begin{definition}
The \emph{pattern of a walk} $(h_1,\dots,h_p)\in V^{p}$ with respect to $\mathfrak{s}^k_{G}(i_1,\dots,i_{k})$ is $(\mathfrak{s}^k_{G}(i_1,\dots,i_{k})_{h_10},\dots,\mathfrak{s}^k_{G}(i_1,\dots,i_{k})_{h_p0})$.
\end{definition}

\begin{theorem}\label{thm:gi:colorfulWalk}
  Let $\mathfrak{s}^k_{G_1}(i_1,\dots,i_{k})=\mathfrak{s}^k_{G_2}(j_1,\dots,j_{k})$ and let $i,j$ be two nodes such that $\mathfrak{s}^k_{G_1}(i_1,\dots,i_{k})_i=\mathfrak{s}^k_{G_2}(j_1,\dots,j_{k})_j$. Then the numbers of walks with any given pattern ending at node $i$ in $G_1$ and at node $j$ in $G_2$ are the same.
\end{theorem}
\begin{proof}
  First observe that $\mathfrak{s}^k_{G_1}(i_1,\dots,i_{k})=\mathfrak{s}^k_{G_2}(j_1,\dots,j_{k})$ means that there exists a bijection $\pi:V_1\to V_2$ for which $\mathfrak{s}^k_{G_1}(i_1,\dots,i_{k})_{i'}=\mathfrak{s}^k_{G_2}(j_1,\dots,j_{k})_{\pi(i')}$ holds for all $i'\in V_1$. Since $\mathfrak{s}^k_{G_1}(i_1,\dots,i_{k})_i=\mathfrak{s}^k_{G_2}(j_1,\dots,j_{k})_j$, one can also prescribe that $\pi(i)=j$.
  We show by induction on the length of the walk that the numbers of walks with any given pattern ending at node $i'$ in $G_1$ and at node $\pi(i')$ in $G_2$ are the same.
  By definition, the number of walks of length zero is the same for any pattern of length one. By induction, assume that the statement holds for walks of length $(l-1)$. Note that $\mathfrak{s}^k_{G_1}(i_1,\dots,i_{k})_{i'}=\mathfrak{s}^k_{G_2}(j_1,\dots,j_{k})_{\pi(i')}$ implies that $\{\mathfrak{s}^k_{G_1}(i_1,\dots,i_k)_{i'',l-1}:i''\in N_{G_1}(i')\}^\#=\{\mathfrak{s}^k_{G_2}(j_1,\dots,j_k)_{j'',l-1}:j''\in N_{G_2}(\pi(i'))\}^\#$. By the induction hypothesis, this means that the numbers of walks of length $(l-1)$ with any given pattern ending at the neighbors of $i'$ in $G_1$ and those ending at the neighbors of $\pi(i')$ in $G_2$ are the same.
  Any walk of length $l\geq 1$ ending at a node $t$ consists of an edge $tt'$ incident to $t$ and a walk of length $(l-1)$ ending at $t'$. Applying this observation for $t=i'$ and $t=j'$, the numbers of walks of length $l$ with any given pattern ending at node $i'$ in $G_1$ and at node $\pi(i')$ in $G_2$ are the same, which had to be shown.
\end{proof}

\begin{definition}
  A subset $X$ of the nodes is a \emph{$k$-separator} if $|X|\leq k$ and $G\setminus X$ consists of more connected components than $G$.
\end{definition}

\begin{theorem}
  If $i\in V_1$ is contained in a $(k-1)$-separator of $G_1$ and $j\in V_2$ is not contained in any $k$-separators of $G_2$, then $\mathfrak{s}^k_{G_1}(i)\neq\mathfrak{s}^k_{G_2}(j)$.
\end{theorem}
\begin{proof}
  By Claim~\ref{cl:gi:distPowerOfSWLIsNonDecr}, one can assume that $|X|=k-1$.
  Let $i_1,\dots,i_{k-1}\in V_1$ denote the nodes of $X$, and let $i_k$ be an arbitrary node in $V_1\setminus X$.
  We can assume that $i=i_1$.
  By contradiction, if $\mathfrak{s}^k_{G_1}(i)=\mathfrak{s}^k_{G_2}(j)$, then there exist $j_1,\dots,j_k\in V_2$ such that $\mathfrak{s}^k_{G_1}(i_1,\dots,i_k)=\mathfrak{s}^k_{G_2}(j_1,\dots,j_k)$, where $j_1=j$.
  Let $C$ be an arbitrary connected component of $G_1\setminus X$ that does not contain $i_k$, and let $i'$ be a node of $C$.
  Observe that any walk in $G_1$ between node $i_k$ and $i'$ crosses at least one of $i_1,\dots,i_{k-1}$.
  At the same time, there exists a walk in $G_2$ between $j_k$ and any other node of $V_2\setminus\{j_1,\dots,j_{k-1}\}$ that avoids all nodes $j_1,\dots,j_{k-1}$, contradicting Theorem~\ref{thm:gi:colorfulWalk}.
\end{proof}
By the previous theorem, one gets that if $\mathfrak{s}^k_{G_1}(i)=\mathfrak{s}^k_{G_2}(j)$, then $i$ is included in a $(k-1)$-separator if and only if $j$ is, hence the following corollary immediately follows.
\begin{corollary}
  If $G_1$ is $k$-connected and $G_2$ is not, then $\mathfrak{s}^k_{G_1}\neq\mathfrak{s}^k_{G_2}$
\end{corollary}

\begin{theorem}\label{thm:gi:complements}
  Let $G_1$, $G_2$ be two graphs, and let $\overline G_1,\overline G_2$ be their complements, respectively. Then $\mathfrak{s}^{k}_{G_1}=\mathfrak{s}^{k}_{G_2}$ if and only if $\mathfrak{s}^{k}_{\overline G_1}=\mathfrak{s}^{k}_{\overline G_2}$.
\end{theorem}
\begin{proof}
  It suffices to show that if $\mathfrak{s}^{k}_{G_1}=\mathfrak{s}^{k}_{G_2}$, then $\mathfrak{s}^{k}_{\overline G_1}=\mathfrak{s}^{k}_{\overline G_2}$, as this immediately implies the other direction, as well.
  Let $i_1,\dots,i_{k}\in V_1$ and $j_1,\dots,j_{k}\in V_2$ be such that $\mathfrak{s}^k_{G_1}(i_1,\dots,i_{k})=\mathfrak{s}^k_{G_2}(j_1,\dots,j_{k})$.

  First, we prove by induction on $l$ that if $i\in V_1$ and $j\in V_2$ are such that $\mathfrak{s}^k_{G_1}(i_1,\dots,i_{k})_{il}=\mathfrak{s}^k_{G_2}(j_1,\dots,j_{k})_{jl}$, then $\mathfrak{s}^k_{\overline G_1}(i_1,\dots,i_{k})_{il}=\mathfrak{s}^k_{\overline G_2}(j_1,\dots,j_{k})_{jl}$ holds.
  The base case, $l=0$, clearly holds, since the initial labels are the same in all four graphs.
  By the induction hypothesis, assume that the statement is true for $(l-1)$ for some $l\geq 1$.
  First observe that $\mathfrak{s}^k_{G_1}(i_1,\dots,i_{k})=\mathfrak{s}^k_{G_2}(j_1,\dots,j_{k})$ implies that
  \begin{equation}\label{eq:gi:setOfAllHashs}
    \{\mathfrak{s}^k_{G_1}(i_1,\dots,i_{k})_{i',l-1} : i'\in V_1\}^\#=\{\mathfrak{s}^k_{G_2}(j_1,\dots,j_{k})_{j',l-1} : j'\in V_2\}^\#.
  \end{equation}
  By definition, $\mathfrak{s}^k_{G_1}(i_1,\dots,i_{k})_{il}=\mathfrak{s}^k_{G_2}(j_1,\dots,j_{k})_{jl}$ implies that
    \begin{equation}\label{eq:gi:setOfNeighHashs}
        \{\mathfrak{s}^k_{G_1}(i_1,\dots,i_k)_{i',l-1}:i'\in N_{G_1}(i)\}^\#=\{\mathfrak{s}^k_{ G_2}(j_1,\dots,j_k)_{j',l-1}:j'\in N_{G_2}(j)\}^\#
    \end{equation}
    Hence,
    \begin{multline}
      \{\mathfrak{s}^k_{G_1}(i_1,\dots,i_k)_{i',l-1}:i'\in N_{\overline G_1}(i)\}^\#\\
      =\{\mathfrak{s}^k_{G_1}(i_1,\dots,i_{k})_{i',l-1} : i'\in V_1\}^\#\setminus \{\mathfrak{s}^k_{G_1}(i_1,\dots,i_k)_{i',l-1}:i'\in N_{G_1}(i)\}^\#\\
      =\{\mathfrak{s}^k_{G_2}(j_1,\dots,j_{k})_{j',l-1} : j'\in V_2\}^\#\setminus
      \{\mathfrak{s}^k_{ G_2}(j_1,\dots,j_k)_{j',l-1}:j'\in N_{G_2}(j)\}^\#\\
      =\{\mathfrak{s}^k_{G_2}(j_1,\dots,j_k)_{j',l-1}:j'\in N_{\overline G_2}(j)\}^\#,
    \end{multline}
    where the second equality holds by (\ref{eq:gi:setOfAllHashs}) and (\ref{eq:gi:setOfNeighHashs}).
    From this, $\{\mathfrak{s}^k_{\overline G_1}(i_1,\dots,i_k)_{i',l-1}:i'\in N_{\overline G_1}(i)\}^\#=\{\mathfrak{s}^k_{\overline G_2}(j_1,\dots,j_k)_{j',l-1}:j'\in N_{\overline G_2}(j)\}^\#$ by induction, hence $\mathfrak{s}^k_{\overline G_1}(i_1,\dots,i_{k})_{il}=\mathfrak{s}^k_{\overline G_2}(j_1,\dots,j_{k})_{jl}$ follows, which proves the statement.

    We show by induction on $q=k,\dots,0$ that if $\mathfrak{s}^k_{G_1}(i_1,\dots,i_{q}) = \mathfrak{s}^k_{G_2}(j_1,\dots,j_{q})$, then $\mathfrak{s}^k_{\overline G_1}(i_1,\dots,i_{q}) = \mathfrak{s}^k_{\overline G_2}(j_1,\dots,j_{q})$ --- which is the statement of the theorem if $q=0$.
    The base case, $q=k$, immediately follows from the above statement.
    By the induction hypothesis, assume that our statement holds for $(q+1)$ and show that it holds for $q$, as well.
    The assumption, $\mathfrak{s}^k_{G_1}(i_1,\dots,i_{q}) = \mathfrak{s}^k_{G_2}(j_1,\dots,j_{q})$, means that $\{\mathfrak{s}^k_{G_1}(i_1,\dots,i_{q},i) : i\in V_1\setminus\{i_1,\dots,i_{q}\}\}^\#= \{\mathfrak{s}^k_{G_2}(j_1,\dots,j_{q},j) : j\in V_2\setminus\{j_1,\dots,j_{q}\}\}^\#$.
    From this, it follows by induction that $\{\mathfrak{s}^k_{\overline G_1}(i_1,\dots,i_{q},i) : i\in V_1\setminus\{i_1,\dots,i_{q}\}\}^\#= \{\mathfrak{s}^k_{\overline G_2}(j_1,\dots,j_{q},j) : j\in V_2\setminus\{j_1,\dots,j_{q}\}\}^\#$, that is, $\mathfrak{s}^k_{\overline G_1}(i_1,\dots,i_{q}) = \mathfrak{s}^k_{\overline G_2}(j_1,\dots,j_{q})$, which had to be shown.
\end{proof}

A graph class $\mathcal{G}$ is \emph{self-complementary} if the complement of any graph of $\mathcal{G}$ is also in $\mathcal{G}$.
The previous theorem immediately implies the following observation, which will be useful in Section~\ref{sec:gi:expRes}.
\begin{corollary}\label{cor:gi:complementReduction}
  Let $\mathcal{G}$ be a self-complementary graph class.
  Then $\mathfrak{s}^k$ identifies all graphs in $\mathcal{G}$ if and only if it identifies all graphs in $\mathcal{G'}:=\{G\in\mathcal{G}:|E_G|\leq\lceil{\binom{|V_G|}{2}}/2\rceil\}$.
\end{corollary}

In what follows, we prove that two 3-connected planar graphs are isomorphic if and only if they are $\mathfrak{s}^3$-equivalent.
\begin{lemma}\label{lem:gi:Labels3ConnPlanar}
  Let $G$ be a 3-connected planar graph.
  If $i_1,i_2,i_3\in V$ are three distinct nodes sharing a common face, then $\mathfrak{s}^3_{G}(i_1,i_2,i_3)_i \neq \mathfrak{s}^3_{G}(i_1,i_2,i_3)_j$ for all distinct $i,j\in V$.
\end{lemma}
\begin{proof}
  For all $k\in\Z_+$, let $\gamma_k$ be a $V\rightarrow\R^2$ function defined as follows.
  If $k=0$, then let
  \begin{equation}
    \gamma_0(i):=
    \begin{cases}
      (0,0)        & \text{if } i=i_1,\\
      (0,1)        & \text{if } i=i_2,\\
      (1,0)        & \text{if } i=i_3,\\
      (1,1)        & \text{otherwise.}\\
    \end{cases}
  \end{equation}
  For $k\geq 1$, let
  \begin{equation}
    \gamma_k(i):=
    \begin{cases}
      \gamma_{k-1}(i)                                                & \text{if } i\in\{i_1,i_2,i_3\},\\
      \frac{1}{\delta_G(i)}\sum\limits_{i'\in N_G(i)}\gamma_{k-1}(i')& \text{otherwise.}\\
    \end{cases}
  \end{equation}

  As $k$ goes to infinity, $\gamma_k$ converges to a planar embedding~\cite{TutteDrawGraphs}, hence $\gamma_k$ is an injection for sufficiently large $k$.
  Therefore, it suffices to show that
  \begin{equation}\label{eq:gi:IndEq}
    \gamma_k(i)\neq\gamma_k(j) \Longrightarrow \mathfrak{s}^3_{G}(i_1,i_2,i_3)_{ik}\neq\mathfrak{s}^3_{G}(i_1,i_2,i_3)_{jk}
  \end{equation}
  holds for all $i,j\in V$, which we prove by induction on $k$.

  The base case, $\gamma_0(i)\neq\gamma_0(j) \Longrightarrow \mathfrak{s}^3_{G}(i_1,i_2,i_3)_{i0}\neq\mathfrak{s}^3_{G}(i_1,i_2,i_3)_{j0}$, easily follows by definition. By induction, suppose that (\ref{eq:gi:IndEq}) holds for $(k-1)$, where $k\geq 1$.

  If $i\in\{i_1,i_2,i_3\}$ or $j\in\{i_1,i_2,i_3\}$, then (\ref{eq:gi:IndEq}) holds, since all the rows of $\restr{\mathfrak{s}^3_{G}(i_1,i_2,i_3)}{k}$ corresponding to nodes $\{i_1,i_2,i_3\}$ are unique. Assume that $i,j\notin\{i_1,i_2,i_3\}$. By definition, $\gamma_{k}(i)\neq\gamma_{k}(j)$ means that
  \begin{equation}\label{eq:gi:differentAvgEq}
    \frac{1}{\delta_G(i)}\sum\limits_{i'\in N_G(i)}\gamma_{k-1}(i')\neq \frac{1}{\delta_G(j)}\sum\limits_{j'\in N_G(j)}\gamma_{k-1}(j').
  \end{equation}

  If $\delta_G(i)\neq\delta_G(j)$, then (\ref{eq:gi:IndEq}) holds by the definition of $\mathfrak{s}^3_{G}(i_1,i_2,i_3)$. Otherwise, (\ref{eq:gi:differentAvgEq}) implies that
  \begin{equation*}
    \{\gamma_{k-1}(i'):i'\in N_G(i)\}^\#\neq\{\gamma_{k-1}(j'):j'\in N_G(j)\}^\#
  \end{equation*}
  which, by induction, means that
  \begin{equation}
    \{\mathfrak{s}^3_{G}(i_1,i_2,i_3)_{i',k\!-\!1} : i'\in N_G(i)\}^\#
    \neq\{\mathfrak{s}^3_{G}(i_1,i_2,i_3)_{j',k\!-\!1} : j'\in N_G(j)\}^\#
  \end{equation}
  holds, and therefore $\mathfrak{s}^3_{G}(i_1,i_2,i_3)_{ik}\neq\mathfrak{s}^3_{G}(i_1,i_2,i_3)_{jk}$.
\end{proof}

\begin{theorem}
Two 3-connected planar graphs, $G_1$ and $G_2$ are isomorphic if and only if $\mathfrak{s}^3_{G_1}=\mathfrak{s}^3_{G_2}$.
\end{theorem}
\begin{proof}
  It suffices to show that if $\mathfrak{s}^3_{G_1}=\mathfrak{s}^3_{G_2}$, then $G_1$ and $G_2$ are isomorphic. Let $i_1,i_2,i_3\in V_1$ be three distinct nodes on a common face in some planar embedding of $G_1$. By definition, $\mathfrak{s}^3_{G_1}=\mathfrak{s}^3_{G_2}$ means that $\{\mathfrak{s}^3_{G_1}(i):i\in V_1\}^\#=\{\mathfrak{s}^3_{G_2}(j):j\in V_2\}^\#$, therefore there exists $j_1\in V_2$ such that $\mathfrak{s}^3_{G_1}(i_1)=\mathfrak{s}^3_{G_2}(j_1)$. Similarly, one gets that there exists $j_2\in V_2$ such that $\mathfrak{s}^3_{G_1}(i_1,i_2)=\mathfrak{s}^3_{G_2}(j_1,j_2)$, and there exists $j_3\in V_2$ such that $\mathfrak{s}^3_{G_1}(i_1,i_2,i_3)\peq\mathfrak{s}^3_{G_2}(j_1,j_2,j_3)$.
  The following claim provides the sought bijection.
  \begin{claim}
    There exists a unique bijection $\pi:V_1\to V_2$ such that $\mathfrak{s}^3_{G_1}(i_1,i_2,i_3)_i\peq\mathfrak{s}^3_{G_2}(j_1,j_2,j_3)_{\pi(i)}$ holds for all $i\in V_1$, and this $\pi$ is edge-preserving.
  \end{claim}
  \begin{proof}
    By Lemma~\ref{lem:gi:Labels3ConnPlanar}, the labels in $G_1$ are unique, that is
    \begin{equation}\label{eq:gi:uniqueLabelsInG1Eq}
      \mathfrak{s}^3_{G_1}(i_1,i_2,i_3)_i\peq\mathfrak{s}^3_{G_1}(i_1,i_2,i_3)_{i'}\Longleftrightarrow i=i'
    \end{equation}
    follows. Since $\mathfrak{s}^3_{G_1}(i_1,i_2,i_3)\peq\mathfrak{s}^3_{G_2}(j_1,j_2,j_3)$, the labels in $G_2$ are unique too, that is, we get that
    \begin{equation}\label{eq:gi:uniqueLabelsInG2Eq}
      \mathfrak{s}^3_{G_2}(j_1,j_2,j_3)_j\peq\mathfrak{s}^3_{G_2}(j_1,j_2,j_3)_{j'}\Longleftrightarrow j=j'.
    \end{equation}
    Given that $\mathfrak{s}^3_{G_1}(i_1,i_2,i_3)\peq\mathfrak{s}^3_{G_2}(j_1,j_2,j_3)$, the unique existence of $\pi$ easily follows from (\ref{eq:gi:uniqueLabelsInG1Eq}) and (\ref{eq:gi:uniqueLabelsInG2Eq}). In order to show that $\pi$ is edge-preserving, observe that (\ref{eq:gi:uniqueLabelsInG1Eq}) and (\ref{eq:gi:uniqueLabelsInG2Eq}) hold even for the first $(n+1)$ columns of matrices $\mathfrak{s}^3_{G_1}(i_1,i_2,i_3)$ and $\mathfrak{s}^3_{G_2}(j_1,j_2,j_3)$ by Claim~\ref{cl:gi:onlyShortSWLMatter}. Accordingly, no two rows turn out to be different in column $(n+2)$. More precisely,
    \begin{equation}\label{eq:gi:iterationNPlusTwo}
      \{\mathfrak{s}^3_{G_1}(i_1,i_2,i_3)_{i',n\!+\!1}: i'\in N_{G_1}(i)\}^\#=\{\mathfrak{s}^3_{G_2}(j_1,j_2,j_3)_{j',n\!+\!1}: j'\in N_{G_2}(\pi(i))\}^\#
    \end{equation}
    hold for all nodes $i\in V_1$. Observe that for all nodes $i\in V_1$
    \begin{equation}\label{eq:gi:edgePresEq}
      \{\pi(i'): i'\in N_{G_1}(i)\}^\#=\{j': j'\in N_{G_2}(\pi(i))\}^\#
    \end{equation}
    follows from (\ref{eq:gi:iterationNPlusTwo}), as the rows of matrices $\restr{\mathfrak{s}^3_{G_1}(i_1,i_2,i_3)}{n+1}$ and $\restr{\mathfrak{s}^3_{G_2}(j_1,j_2,j_3)}{n+1}$ uniquely identify the corresponding nodes. Equation~(\ref{eq:gi:edgePresEq}) means that $\pi$ is edge-preserving, which completes the proof.
  \end{proof}
\end{proof}

\subsection{Experimental Results}\label{sec:gi:expRes}
We verified that $\mathfrak{s}^{2}$ identifies all \num{165078520805} graphs on at most 12 nodes, that is, two such graphs are isomorphic if and only if they are $\mathfrak{s}^2$-equivalent.
We also considered all $r$-regular graphs for $r\in\{3,\dots,10\}$ on larger node sets (see Table~\ref{tbl:gi:expResRegular}), and all of them were identified by $\mathfrak{s}^{2}$. It also identifies all \num{43753} strongly regular graphs on at most 64 nodes.

The experiments were run on the HPC called Atlasz~\cite{atlasz}, which is a computer cluster with $45$ computing nodes, each of which is equipped with an 18-core Intel Xeon Gold 6240 CPU and 90GB of RAM.
Essentially, this means that the system consists of $45$ separate computers, which can communicate with each other over network connection.
We computed the hash values of all the considered graphs using these computing nodes, each processing the graphs in parallel.
The small graphs were generated using the Nauty package~\cite{NautyII}.
Note that the graphs on $n$ nodes form a self-complementary graph class, hence one can significantly reduce the number of graphs to be investigated by Corollary~\ref{cor:gi:complementReduction}.
We collected the hash values on hard disk --- as storing about four terabytes of data in memory was not possible.
To minimize the running time and the space requirement, we used a non-cryptographic 64-bit hash function, called MurmurHash~\cite{MurmurHash}.
We found that there were only a couple of thousand graph pairs with the same hash value, for which we could easily compute the much stronger and slower SHA512 hash values, which successfully distinguished all the graphs.
The computation took more than three weeks.

\begin{table}[t]
\centering
\begin{tabular}{|r|c|r|}
\hline
\multicolumn{1}{|c|}{degree (r)} & \multicolumn{1}{c|}{number of nodes (n)} & \multicolumn{1}{c|}{\parbox{3.6cm}{\centering\vspace{1mm}number of $r$-regular graphs\\on at most $n$ nodes\vspace{1mm}}} \\ \hline
3                       & $\leq 26$              & \num{2220297316}                            \\ \hline
4                       & $\leq 18$              & \num{1081035905}                            \\ \hline
5                       & $\leq 16$              & \num{2588603970}                            \\ \hline
6                       & $\leq 15$              & \num{1492278977}                            \\ \hline
7                       & $\leq 14$              & \num{21610854}                              \\ \hline
8                       & $\leq 15$              & \num{1473763950}                            \\ \hline
9                       & $\leq 14$              & \num{88203}                                 \\ \hline
10                      & $\leq 16$              & \num{2585942872}                            \\ \hline
\end{tabular}
\caption{ Maximum node numbers of the tested regular graphs, and the number of such graphs. }
\label{tbl:gi:expResRegular}
\end{table}

\subsection{Indistinguishable Graph Pairs}
In the light of the positive results and the computational experiments, it is quite natural to ask whether there exists a non-isomorphic graph pair that can not be distinguished by $\mathfrak{s}^k$.
In this section, we construct non-isomorphic graph pairs that have the same $\mathfrak{s}^1$ or $\mathfrak{s}^2$ fingerprints.

We need the following well-known definition.
\begin{definition}
A regular graph $G=(V, E)$ with $n$ nodes and degree $d$ is said to be \emph{strongly regular} if there exist integers $\lambda$ and $\mu$ such that every two adjacent nodes have exactly $\lambda$ common neighbors, and every two non-adjacent nodes have exactly $\mu$ common neighbors.
\end{definition}

\begin{theorem}
  If $G_1$ and $G_2$ are connected strongly regular graphs with the same parameters $(n,d,\lambda,\mu)$, then $\mathfrak{s}^1_{G_1}=\mathfrak{s}^1_{G_2}$.
\end{theorem}
\begin{proof}
  Recall that $\mathfrak{s}^1_{G_1}=\{\mathfrak{s}^1_{G_1}(i) : i\in V_1\}$ and $\mathfrak{s}^1_{G_2}=\{\mathfrak{s}^1_{G_2}(j) : j\in V_1\}$. It suffices to show that $\mathfrak{s}^1_{G_1}(i)=\mathfrak{s}^1_{G_2}(j)$ for all $i\in V_1,j\in V_2$. Let $i\in V_1,j\in V_2$ be arbitrary. First observe that any node of $G_1$ and $G_2$ can be reached in at most two steps from $i$ and $j$, respectively. Let $L^1_q\subseteq V_1$ and $L^2_q\subseteq V_2$ denote the nodes that are at distance $q$ from $i$ in $G_1$ and from $j$ in $G_2$, respectively, where $q\in\{0,1,2\}$.

  By definition, $L^1_0=\{i\}$, $L^2_0=\{j\}$, $L^1_1=N_{G_1}(i)$, $L^2_1=N_{G_2}(j)$, $L^1_2=V_1\setminus(L^1_0\cup L^1_1)$ and $L^2_2=V_2\setminus(L^2_0\cup L^2_1)$. Observe that there are no edges between $L^p_0$ and $L^p_2$ for $p\in\{1,2\}$, hence any node of $L^p_2$ has exactly $\mu$ neighbors in $L^p_1$ and exactly $(d-\mu)$ neighbors in $L^p_2$.

  In what follows, it is shown that $\restr{\mathfrak{s}^1_{G_1}(i)}{l} = \restr{\mathfrak{s}^1_{G_2}(j)}{l}$ for all $l\geq 0$. If $l=0$, then the statement holds because the initial labels are the same. For $l=1$, the labels of the nodes of $L^p_0,L^p_1$ and $L^p_2$ are, $(1,\{\underbrace{0,\dots,0}_d\}^\#)$, $(0,\{1,\underbrace{0,\dots,0}_{d-1}\}^\#)$ and $(0,\{\underbrace{0,\dots,0}_{d}\}^\#)$, respectively, that is, they are the same in both graphs.

  For $l\geq 2$, we show that the labels of the nodes in $L^1_q$ and $L^2_q$ remain the same in $G_1$ and $G_2$ for all $q\in\{0,1,2\}$. Let $h_q$ denote the labels for $(l-1)$ of the nodes in $L^1_q$ for $q\in\{0,1,2\}$. The new labels of the nodes of $L^p_0,L^p_1$ and $L^p_2$ are, by definition, $(h_0,\{\underbrace{h_1,\dots,h_1}_d\}^\#)$, $(h_1,\{h_0,\underbrace{h_1,\dots,h_1}_{\lambda},\underbrace{h_2,\dots,h_2}_{d-\lambda-1}\}^\#)$ and $(h_2,\{\underbrace{h_1,\dots,h_1}_{\mu},\underbrace{h_2,\dots,h_2}_{d-\mu}\}^\#)$, respectively, for both $p=1$ and $p=2$. Hence one gets that $\mathfrak{s}^1_{G_1}(i)=\mathfrak{s}^1_{G_2}(j)$ holds for all $i\in V_1$, $j\in V_2$, meaning that $\mathfrak{s}^1_{G_1}=\mathfrak{s}^1_{G_2}$, which completes the proof.
\end{proof}

As there exist two non-isomorphic strongly regular graphs with the same parameters (the smallest two such graphs have parameters $(16,6,2,2)$), the previous theorem immediately implies the following.

\begin{corollary}
  There exist two non-isomorphic graphs $G_1$ and $G_2$ which are not distinguished by $\mathfrak{s}^1$, that is, $\mathfrak{s}^1_{G_1}=\mathfrak{s}^1_{G_2}$.
\end{corollary}

A long-standing question of the authors is whether $\mathfrak{s}^{2}$ identifies all graphs. Now, two non-isomorphic graphs are presented that have the same $\mathfrak{s}^{2}$ fingerprints.

\begin{theorem}\label{thm:gi:2potatoCntExample}
  There exist two non-isomorphic graphs $G_1,G_2$ such that $\mathfrak{s}^{2}_{G_1}=\mathfrak{s}^{2}_{G_2}$.
\end{theorem}
\begin{proof}[Sketch of the proof]
  Let $G=(V,E)$ denote the strongly regular graph with parameters $(35,18,9,9)$ given in Appendix~\ref{app:gi:2potatoCntExample}, and let $u,v$ be the two nodes corresponding to the two highlighted columns. Let $G'=(V',E')$ and $G''=(E'',V'')$ be two disjoint copies of $G$, and let $u',v'\in V'$ denote the copies of $u,v\in V$ in $G'$ and $u'',v''\in V''$ denote the copies of $u,v\in V$ in $G''$, respectively. Consider the following construction. First, let $G_1$ denote the graph obtained by unifying 1) the two copies of $u$ and 2) the two copies of $v$. Second, let $G_2$ denote the graph obtained by unifying 1) nodes $u\in V'$ and $v\in V''$, and 2) unifying nodes $v\in V'$ and $u\in V''$. Figure~\ref{fig:gi:2potatoCntExample} illustrates the construction. We verify by computer that $\mathfrak{s}^{2}_{G_1}=\mathfrak{s}^{2}_{G_2}$ and $\mathfrak{s}^{3}_{G_1}\neq\mathfrak{s}^{3}_{G_2}$. The latter implies that the two graphs are not isomorphic, hence $\mathfrak{s}^{2}$ does not distinguish graphs $G_1$ and $G_2$.
\end{proof}

\begin{remark}
Let $G=(V,E)$ be the graph given in Appendix~\ref{app:gi:2potatoCntExample}. The two highlighted nodes of $G$ given by Appendix~\ref{app:gi:2potatoCntExample}, $u$ and $v$, were selected such that
\[\{\mathfrak{s}^{2}_{G}(i):i\in N_G(u)\}^\#\neq \{\mathfrak{s}^{2}_{G}(i):i\in N_G(v)\}^\#\]
holds.
However, this condition is not sufficient to provide a counterexample.
\end{remark}

\def\ra{2}
\def\d{1.7}
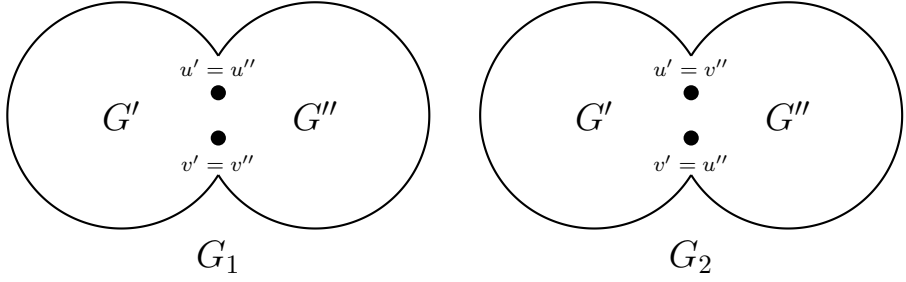
\begin{figure}
  \centering
  \begin{subfigure}[t]{.49\textwidth}
    \centering
    \begin{tikzpicture}[
       scale=.75,
       thick,
       acteur/.style={
         circle,
         fill=black,
         thick,
         inner sep=2pt,
         minimum size=0.2cm
       }
     ]
    \draw [samples=1000,domain={acos(\d/2)}:{360-acos(\d/2)}] plot ({2*(-\d+cos(\x))}, {2*sin(\x)});
    \draw [samples=1000,domain={acos(\d/2)}:{360-acos(\d/2)}] plot ({-2*cos(\x)}, {2*sin(\x)});
       \node (n1) at (-\d,0.4) [acteur,label={$u'=u''$}]{};
       \node (n2) at (-\d,-0.4) [acteur,label=below:{$v'=v''$}]{};
    \node [font=\bfseries] at (0,0) {\Large $G''$};
    \node [font=\bfseries] at (-2*\d,0) {\Large $G'$};
    \node [font=\bfseries] at (-\d,-2.5) {\Large $G_1$};
    \end{tikzpicture}
  \end{subfigure}
  \hfill
  \begin{subfigure}[t]{.49\textwidth}
      \centering
      \begin{tikzpicture}[
        scale=.75,
       thick,
       acteur/.style={
         circle,
         fill=black,
         thick,
         inner sep=2pt,
         minimum size=0.2cm
       }
     ]
    \draw [samples=1000,domain={acos(\d/2)}:{360-acos(\d/2)}] plot ({2*(-\d+cos(\x))}, {2*sin(\x)});
    \draw [samples=1000,domain={acos(\d/2)}:{360-acos(\d/2)}] plot ({-2*cos(\x)}, {2*sin(\x)});
       \node (n1) at (-\d,0.4) [acteur,label={$u'=v''$}]{};
       \node (n2) at (-\d,-0.4) [acteur,label=below:{$v'=u''$}]{};
    \node [font=\bfseries] at (0,0) {\Large $G''$};
    \node [font=\bfseries] at (-2*\d,0) {\Large $G'$};
    \node [font=\bfseries] at (-\d,-2.5) {\Large $G_2$};
\end{tikzpicture}
  \end{subfigure}
  \caption{Illustration of the construction described in the proof of Theorem~\ref{thm:gi:2potatoCntExample}.}\label{fig:gi:2potatoCntExample}
\end{figure}

\subsection{An Alternative Way of Aggregation}
We conclude the paper with a natural alternative to Notation~\ref{nota:gi:potato}, which is compared to $\mathfrak{s}^k$.

\begin{notation}\label{nota:gi:potatoB}
For $q=k,\dots,0$ and $i\in V$, let
\begin{equation*}
\mathfrak{t}^k_G(i_1,\dots,i_{q};i):=
\begin{cases}
    \mathfrak{s}^k_G(i_1,\dots,i_{k})_i                                                & \text{if } q=k,\\
    \{\mathfrak{t}^k_G(i_1,\dots,i_{q},j;i) : j\in V\setminus\{i,i_1,\dots,i_{q}\}\}^\#& \text{otherwise,}
\end{cases}
\end{equation*}
\end{notation}
and let
\begin{equation}
    \mathfrak{t}^k_G:=\{\mathfrak{t}^k_G(;i) : i\in V\}^\#.
\end{equation}

Based on the following theorem, the two ways of aggregation, $\mathfrak{s}^{k}_{G}$ and $\mathfrak{t}^{k}_{G}$ given by Notations~\ref{nota:gi:potato}~and~\ref{nota:gi:potatoB}, are strongly related.
\begin{theorem}\label{thm:gi:twoHashs}
  If $\mathfrak{s}^{k+1}_{G_1}=\mathfrak{s}^{k+1}_{G_2}$, then $\mathfrak{t}^k_{G_1}=\mathfrak{t}^k_{G_2}$. Similarly, if $\mathfrak{t}^{k+1}_{G_1}=\mathfrak{t}^{k+1}_{G_2}$, then $\mathfrak{s}^k_{G_1}=\mathfrak{s}^k_{G_2}$.
\end{theorem}
\begin{proof}
  We deal with the first part only, as the second one can be proved in a similar way. Let $i\in V_1$, $j\in V_2$ be such that $\mathfrak{s}^{k+1}_{G_1}(i)=\mathfrak{s}^{k+1}_{G_2}(j)$. First we show that
  \begin{equation}\label{eq:gi:subStm}
    \mathfrak{s}^{k+1}_{G_1}(i,i_1,\dots,i_q)=\mathfrak{s}^{k+1}_{G_2}(j,j_1,\dots,j_q) \Longrightarrow \mathfrak{t}^k_{G_1}(i_1,\dots,i_q;i)=\mathfrak{t}^k_{G_2}(j_1,\dots,j_q;j)
  \end{equation}
  holds for all $i_1,\dots,i_q\in V_1$ and $j_1,\dots,j_q\in V_2$ and for all $q=k,\dots,0$. The proof is by induction on $q=k,\dots,0$. In the base case, when $q=k$, if $\mathfrak{t}^k_{G_1}(i_1,\dots,i_q;i)\neq\mathfrak{t}^k_{G_2}(j_1,\dots,j_q;j)$, then $\mathfrak{s}^{k+1}_{G_1}(i,i_1,\dots,i_q)\neq\mathfrak{s}^{k+1}_{G_2}(j,j_1,\dots,j_q)$, since the labels of the nodes in the latter case refine the former ones. (Note that this is not necessarily the case if we considered $\mathfrak{s}^{k}$ instead of $\mathfrak{s}^{k+1}$!)
  For $q<k$, assume that the statement holds for all larger values of $q$ by the induction hypothesis. By definition, $\mathfrak{s}^{k+1}_{G_1}(i,i_1,\dots,i_q)=\mathfrak{s}^{k+1}_{G_2}(j,j_1,\dots,j_q)$ means that \begin{multline}\label{eq:gi:hashMeaning1}
    \{\mathfrak{s}^{k+1}_{G_1}(i,i_1,\dots,i_{q},i') : i'\in V_1\setminus\{i,i_1,\dots,i_{q}\}\}^\#\\
    =\{\mathfrak{s}^k_{G_2}(j,j_1,\dots,j_{q},j') : j'\in V\setminus\{j,j_1,\dots,j_{q}\}\}^\#,
  \end{multline}
  whereas $\mathfrak{t}^k_{G_1}(i_1,\dots,i_q;i)=\mathfrak{t}^k_{G_2}(j_1,\dots,j_q;j)$ means that
  \begin{multline}\label{eq:gi:hashMeaning2}
    \{\mathfrak{t}^k_{G_1}(i_1,\dots,i_{q},i';i) : i'\in V_1\setminus\{i,i_1,\dots,i_{q}\}\}^\#\\
    =\{\mathfrak{t}^k_{G_2}(j_1,\dots,j_{q},j';j) : j'\in V\setminus\{j,j_1,\dots,j_{q}\}\}^\#.
  \end{multline}
  By the induction hypothesis, (\ref{eq:gi:hashMeaning1}) implies (\ref{eq:gi:hashMeaning2}), which proves (\ref{eq:gi:subStm}).

  To complete the proof of the theorem, observe that for $q=0$, (\ref{eq:gi:subStm}) means that $\mathfrak{t}^k_{G_1}(;i)=\mathfrak{t}^k_{G_2}(;j)$ holds if $\mathfrak{s}^{k+1}_{G_1}(i)=\mathfrak{s}^{k+1}_{G_2}(j)$, which --- given that $\mathfrak{s}^{k+1}_{G_1}=\mathfrak{s}^{k+1}_{G_2}$ --- implies $\mathfrak{t}^k_{G_1}=\mathfrak{t}^k_{G_2}$.
\end{proof}

\begin{appendices}
\section{Graph for the proof of Theorem~\ref{thm:gi:2potatoCntExample}}\label{app:gi:2potatoCntExample}
\renewcommand{\arraystretch}{1.2}
\setlength{\arrayrulewidth}{0pt}
\begin{table}[H]
\centering
\scalebox{0.5}[.65]{
\begin{tabular}{ b | b | b | b | b | b | b | b | b | b | b | b | b | b | b | b | b | b | b | b | b | b | b | b | a | a | b | b | b | b | b | b | b | b | b | b | b |}

\mc{1}{} & \mc{1}{} & \mc{1}{} & \mc{1}{} & \mc{1}{} & \mc{1}{} & \mc{1}{} & \mc{1}{} & \mc{1}{} & \mc{1}{} & \mc{1}{} & \mc{1}{} & \mc{1}{} & \mc{1}{} & \mc{1}{} & \mc{1}{} & \mc{1}{} & \mc{1}{} & \mc{1}{} & \mc{1}{} & \mc{1}{} & \mc{1}{} & \mc{1}{} & \mc{1}{} & \mc{1}{$u$} & \mc{1}{$v$} & \mc{1}{} & \mc{1}{} & \mc{1}{} & \mc{1}{} & \mc{1}{} & \mc{1}{} & \mc{1}{} & \mc{1}{} & \mc{1}{} \\

\cline{2-36}
 & 0 & 1 & 1 & 1 & 1 & 1 & 1 & 1 & 1 & 1 & 1 & 1 & 1 & 1 & 1 & 1 & 1 & 1 & 1 & 0 & 0 & 0 & 0 & 0 & 0 & 0 & 0 & 0 & 0 & 0 & 0 & 0 & 0 & 0 & 0 \\\cline{2-36}
 & 1 & 0 & 1 & 1 & 1 & 1 & 1 & 1 & 1 & 1 & 1 & 0 & 0 & 0 & 0 & 0 & 0 & 0 & 0 & 1 & 1 & 1 & 1 & 1 & 1 & 1 & 1 & 0 & 0 & 0 & 0 & 0 & 0 & 0 & 0 \\\cline{2-36}
 & 1 & 1 & 0 & 1 & 1 & 1 & 1 & 1 & 1 & 1 & 1 & 0 & 0 & 0 & 0 & 0 & 0 & 0 & 0 & 0 & 0 & 0 & 0 & 0 & 0 & 0 & 0 & 1 & 1 & 1 & 1 & 1 & 1 & 1 & 1 \\\cline{2-36}
 & 1 & 1 & 1 & 0 & 1 & 1 & 1 & 0 & 0 & 0 & 0 & 1 & 1 & 1 & 1 & 0 & 0 & 0 & 0 & 1 & 1 & 1 & 1 & 0 & 0 & 0 & 0 & 1 & 1 & 1 & 1 & 0 & 0 & 0 & 0 \\\cline{2-36}
 & 1 & 1 & 1 & 1 & 0 & 0 & 0 & 1 & 1 & 0 & 0 & 1 & 1 & 1 & 0 & 1 & 0 & 0 & 0 & 1 & 1 & 0 & 0 & 1 & 1 & 0 & 0 & 1 & 0 & 0 & 0 & 1 & 1 & 1 & 0 \\\cline{2-36}
 & 1 & 1 & 1 & 1 & 0 & 0 & 0 & 1 & 0 & 1 & 0 & 1 & 1 & 0 & 0 & 1 & 1 & 0 & 0 & 0 & 0 & 1 & 1 & 0 & 0 & 1 & 1 & 0 & 1 & 1 & 0 & 1 & 0 & 0 & 1 \\\cline{2-36}
 & 1 & 1 & 1 & 1 & 0 & 0 & 0 & 0 & 1 & 1 & 0 & 0 & 0 & 0 & 1 & 0 & 1 & 1 & 1 & 1 & 1 & 0 & 0 & 1 & 1 & 0 & 0 & 0 & 1 & 1 & 1 & 0 & 0 & 0 & 1 \\\cline{2-36}
 & 1 & 1 & 1 & 0 & 1 & 1 & 0 & 0 & 0 & 0 & 1 & 1 & 1 & 0 & 1 & 0 & 1 & 0 & 0 & 0 & 0 & 0 & 0 & 1 & 1 & 1 & 1 & 0 & 0 & 0 & 1 & 0 & 1 & 1 & 1 \\\cline{2-36}
 & 1 & 1 & 1 & 0 & 1 & 0 & 1 & 0 & 0 & 0 & 1 & 0 & 0 & 0 & 0 & 1 & 1 & 1 & 1 & 1 & 1 & 1 & 1 & 0 & 0 & 0 & 0 & 0 & 0 & 0 & 0 & 1 & 1 & 1 & 1 \\\cline{2-36}
 & 1 & 1 & 1 & 0 & 0 & 1 & 1 & 0 & 0 & 0 & 1 & 0 & 0 & 1 & 0 & 1 & 0 & 1 & 1 & 0 & 0 & 0 & 0 & 1 & 1 & 1 & 1 & 1 & 1 & 1 & 0 & 1 & 0 & 0 & 0 \\\cline{2-36}
 & 1 & 1 & 1 & 0 & 0 & 0 & 0 & 1 & 1 & 1 & 0 & 0 & 0 & 1 & 1 & 0 & 0 & 1 & 1 & 0 & 0 & 1 & 1 & 0 & 0 & 1 & 1 & 1 & 0 & 0 & 1 & 0 & 1 & 1 & 0 \\\cline{2-36}
 & 1 & 0 & 0 & 1 & 1 & 1 & 0 & 1 & 0 & 0 & 0 & 0 & 0 & 1 & 1 & 1 & 1 & 1 & 0 & 1 & 0 & 1 & 0 & 1 & 0 & 1 & 0 & 1 & 1 & 0 & 0 & 0 & 1 & 0 & 1 \\\cline{2-36}
 & 1 & 0 & 0 & 1 & 1 & 1 & 0 & 1 & 0 & 0 & 0 & 0 & 0 & 1 & 1 & 1 & 1 & 0 & 1 & 0 & 1 & 0 & 1 & 0 & 1 & 0 & 1 & 0 & 0 & 1 & 1 & 1 & 0 & 1 & 0 \\\cline{2-36}
 & 1 & 0 & 0 & 1 & 1 & 0 & 0 & 0 & 0 & 1 & 1 & 1 & 1 & 0 & 1 & 0 & 0 & 1 & 1 & 1 & 0 & 1 & 0 & 0 & 1 & 0 & 1 & 1 & 0 & 1 & 0 & 1 & 1 & 0 & 0 \\\cline{2-36}
 & 1 & 0 & 0 & 1 & 0 & 0 & 1 & 1 & 0 & 0 & 1 & 1 & 1 & 1 & 0 & 0 & 0 & 1 & 1 & 1 & 0 & 0 & 1 & 1 & 0 & 0 & 1 & 0 & 1 & 0 & 1 & 0 & 0 & 1 & 1 \\\cline{2-36}
 & 1 & 0 & 0 & 0 & 1 & 1 & 0 & 0 & 1 & 1 & 0 & 1 & 1 & 0 & 0 & 0 & 1 & 1 & 1 & 0 & 1 & 0 & 1 & 1 & 0 & 1 & 0 & 1 & 1 & 0 & 0 & 1 & 0 & 1 & 0 \\\cline{2-36}
 & 1 & 0 & 0 & 0 & 0 & 1 & 1 & 1 & 1 & 0 & 0 & 1 & 1 & 0 & 0 & 1 & 0 & 1 & 1 & 0 & 1 & 1 & 0 & 0 & 1 & 1 & 0 & 0 & 0 & 1 & 1 & 0 & 1 & 0 & 1 \\\cline{2-36}
 & 1 & 0 & 0 & 0 & 0 & 0 & 1 & 0 & 1 & 1 & 1 & 1 & 0 & 1 & 1 & 1 & 1 & 0 & 0 & 1 & 0 & 0 & 1 & 0 & 1 & 1 & 0 & 0 & 1 & 0 & 1 & 1 & 1 & 0 & 0 \\\cline{2-36}
 & 1 & 0 & 0 & 0 & 0 & 0 & 1 & 0 & 1 & 1 & 1 & 0 & 1 & 1 & 1 & 1 & 1 & 0 & 0 & 0 & 1 & 1 & 0 & 1 & 0 & 0 & 1 & 1 & 0 & 1 & 0 & 0 & 0 & 1 & 1 \\\cline{2-36}
 & 0 & 1 & 0 & 1 & 1 & 0 & 1 & 0 & 1 & 0 & 0 & 1 & 0 & 1 & 1 & 0 & 0 & 1 & 0 & 0 & 0 & 1 & 1 & 1 & 1 & 1 & 0 & 0 & 0 & 1 & 0 & 1 & 0 & 1 & 1 \\\cline{2-36}
 & 0 & 1 & 0 & 1 & 1 & 0 & 1 & 0 & 1 & 0 & 0 & 0 & 1 & 0 & 0 & 1 & 1 & 0 & 1 & 0 & 0 & 1 & 1 & 1 & 1 & 0 & 1 & 1 & 1 & 0 & 1 & 0 & 1 & 0 & 0 \\\cline{2-36}
 & 0 & 1 & 0 & 1 & 0 & 1 & 0 & 0 & 1 & 0 & 1 & 1 & 0 & 1 & 0 & 0 & 1 & 0 & 1 & 1 & 1 & 0 & 1 & 0 & 0 & 1 & 1 & 1 & 0 & 1 & 0 & 0 & 1 & 0 & 1 \\\cline{2-36}
 & 0 & 1 & 0 & 1 & 0 & 1 & 0 & 0 & 1 & 0 & 1 & 0 & 1 & 0 & 1 & 1 & 0 & 1 & 0 & 1 & 1 & 1 & 0 & 0 & 0 & 1 & 1 & 0 & 1 & 0 & 1 & 1 & 0 & 1 & 0 \\\cline{2-36}
\rowcolor{Gray}
\cellcolor{white}$u$ & 0 & 1 & 0 & 0 & 1 & 0 & 1 & 1 & 0 & 1 & 0 & 1 & 0 & 0 & 1 & 1 & 0 & 0 & 1 & 1 & 1 & 0 & 0 & 0 & 1 & 1 & 1 & 1 & 1 & 0 & 0 & 0 & 0 & 1 & 1 \\\cline{2-36}
\rowcolor{Gray}
\cellcolor{white}$v$ & 0 & 1 & 0 & 0 & 1 & 0 & 1 & 1 & 0 & 1 & 0 & 0 & 1 & 1 & 0 & 0 & 1 & 1 & 0 & 1 & 1 & 0 & 0 & 1 & 0 & 1 & 1 & 0 & 0 & 1 & 1 & 1 & 1 & 0 & 0 \\\cline{2-36}
 & 0 & 1 & 0 & 0 & 0 & 1 & 0 & 1 & 0 & 1 & 1 & 1 & 0 & 0 & 0 & 1 & 1 & 1 & 0 & 1 & 0 & 1 & 1 & 1 & 1 & 0 & 0 & 1 & 0 & 1 & 1 & 0 & 0 & 1 & 0 \\\cline{2-36}
 & 0 & 1 & 0 & 0 & 0 & 1 & 0 & 1 & 0 & 1 & 1 & 0 & 1 & 1 & 1 & 0 & 0 & 0 & 1 & 0 & 1 & 1 & 1 & 1 & 1 & 0 & 0 & 0 & 1 & 0 & 0 & 1 & 1 & 0 & 1 \\\cline{2-36}
 & 0 & 0 & 1 & 1 & 1 & 0 & 0 & 0 & 0 & 1 & 1 & 1 & 0 & 1 & 0 & 1 & 0 & 0 & 1 & 0 & 1 & 1 & 0 & 1 & 0 & 1 & 0 & 0 & 1 & 1 & 1 & 0 & 1 & 1 & 0 \\\cline{2-36}
 & 0 & 0 & 1 & 1 & 0 & 1 & 1 & 0 & 0 & 1 & 0 & 1 & 0 & 0 & 1 & 1 & 0 & 1 & 0 & 0 & 1 & 0 & 1 & 1 & 0 & 0 & 1 & 1 & 0 & 0 & 1 & 1 & 1 & 0 & 1 \\\cline{2-36}
 & 0 & 0 & 1 & 1 & 0 & 1 & 1 & 0 & 0 & 1 & 0 & 0 & 1 & 1 & 0 & 0 & 1 & 0 & 1 & 1 & 0 & 1 & 0 & 0 & 1 & 1 & 0 & 1 & 0 & 0 & 1 & 1 & 0 & 1 & 1 \\\cline{2-36}
 & 0 & 0 & 1 & 1 & 0 & 0 & 1 & 1 & 0 & 0 & 1 & 0 & 1 & 0 & 1 & 0 & 1 & 1 & 0 & 0 & 1 & 0 & 1 & 0 & 1 & 1 & 0 & 1 & 1 & 1 & 0 & 0 & 1 & 1 & 0 \\\cline{2-36}
 & 0 & 0 & 1 & 0 & 1 & 1 & 0 & 0 & 1 & 1 & 0 & 0 & 1 & 1 & 0 & 1 & 0 & 1 & 0 & 1 & 0 & 0 & 1 & 0 & 1 & 0 & 1 & 0 & 1 & 1 & 0 & 0 & 1 & 1 & 1 \\\cline{2-36}
 & 0 & 0 & 1 & 0 & 1 & 0 & 0 & 1 & 1 & 0 & 1 & 1 & 0 & 1 & 0 & 0 & 1 & 1 & 0 & 0 & 1 & 1 & 0 & 0 & 1 & 0 & 1 & 1 & 1 & 0 & 1 & 1 & 0 & 0 & 1 \\\cline{2-36}
 & 0 & 0 & 1 & 0 & 1 & 0 & 0 & 1 & 1 & 0 & 1 & 0 & 1 & 0 & 1 & 1 & 0 & 0 & 1 & 1 & 0 & 0 & 1 & 1 & 0 & 1 & 0 & 1 & 0 & 1 & 1 & 1 & 0 & 0 & 1 \\\cline{2-36}
 & 0 & 0 & 1 & 0 & 0 & 1 & 1 & 1 & 1 & 0 & 0 & 1 & 0 & 0 & 1 & 0 & 1 & 0 & 1 & 1 & 0 & 1 & 0 & 1 & 0 & 0 & 1 & 0 & 1 & 1 & 0 & 1 & 1 & 1 & 0 \\\cline{2-36}

\end{tabular}
}
\caption{ The adjacency matrix of the graph for the proof of Theorem~\ref{thm:gi:2potatoCntExample}.}\label{table:gi:2potatoCntExample}
\end{table}
\end{appendices}

\bibliography{bibliography}
\bibliographystyle{unsrt}
\end{document}